\documentclass[11pt]{article}

\usepackage[margin=1in]{geometry}

\usepackage[T1]{fontenc}
\usepackage[utf8]{inputenc}
\usepackage[numbers]{natbib}
\usepackage{amsmath,amssymb,amsthm,mathtools}

\usepackage{array}
\usepackage{booktabs}

\usepackage{algorithm}
\usepackage{algpseudocode}
\usepackage{float}

\usepackage{tikz}
\usepackage{mdframed}

\usepackage[dvipsnames]{xcolor}
\usepackage[colorlinks=true,citecolor=blue,linkcolor=blue,urlcolor=blue]{hyperref}

\numberwithin{equation}{section}

\theoremstyle{plain}
\newtheorem{theorem}{Theorem}
\newtheorem{claim}{Claim}[section]
\newtheorem{lemma}{Lemma}[section]
\newtheorem{definition}{Definition}[section]
\newtheorem{corollary}{Corollary}[section]

\theoremstyle{definition}

\newcommand{\OPT}{\operatorname{OPT}}
\newcommand{\cost}{\operatorname{cost}}

\newcommand{\ind}{\mathbf{1}}
\newcommand{\ceil}[1]{\left\lceil #1\right\rceil}
\newcommand{\floor}[1]{\left\lfloor #1\right\rfloor}

\newcommand{\calB}{\mathcal{B}}

\newcommand{\doi}[1]{\href{https://doi.org/#1}{\nolinkurl{#1}}}

\newcommand{\IRA}{\mathsf{IRA}}
\newcommand{\KIRA}{\mathsf{KIRA}}
\newcommand{\LARA}{\mathsf{LARA}}
\newcommand{\KLARA}{\mathsf{KLARA}}
\newcommand{\SARA}{\mathsf{SARA}}

\definecolor{darkgreen}{RGB}{0,100,0}

\title{Improved Approximations for Coflow Scheduling via Iterated Rounding and Stable Matchings}

\author{Mikhail Cherniavskii,
Arseni Fedorovich,
Mikhail Koshelev,
Andrey Kupavskii,\\
Stanislav Moiseev,
Fedor Ozhegov,
Alexander Shekhovtsov
}
\date{}

\begin{document}

\maketitle

\begin{abstract}
We study Coflow Scheduling, introduced by
\hyperlink{cite.Chowdhury2014}{Chowdhury, Zhong, and Stoica (2014)}.
We give deterministic polynomial-time algorithms achieving a
$2.77$-approximation without release dates and a $3.17$-approximation
with integer release dates. These improve the $3.42$ and $4.36$
guarantees of
\hyperlink{cite.RS2025}{Rohwedder and Schnaars (2025)}.

Our approach combines two ideas: (1) a novel iterated rounding scheme, and (2) Galvin's list edge coloring method, which is based on the Gale--Shapley stable matching algorithm.
\end{abstract}

\section{Introduction}\label{sec:introduction}
Coflow scheduling, introduced by Chowdhury, Zhong, and
Stoica~\cite{Chowdhury2014}, models collections of data transfers
with a common completion requirement.
Examples motivating the coflow abstraction include MapReduce
shuffles~\cite{Chowdhury2014}, broadcast and aggregation stages
in Spark~\cite{Chowdhury2012},
and message exchanges within a Pregel
superstep~\cite{Malewicz2010,Chowdhury2012}.
A coflow completes only when all its transfers have finished, so
optimizing individual transfer need not optimize the completion of
the corresponding communication stage.
The Varys system demonstrated that coordinating transfers at the
coflow level can improve completion times compared with per-flow
mechanisms~\cite{Chowdhury2014}.

We study the offline, discrete unit-transfer model: input and output
ports form the two sides of a bipartite multigraph, each unit transfer
is an edge, and each time slot processes a matching. The edges are
partitioned into coflows with positive weights, and the objective is
the weighted sum of their completion times. The release-date variant
additionally requires each edge to be scheduled after the release date
of its coflow.
A concurrent-open-shop relaxation treats each port as a machine and
each coflow as a job. The job's processing requirement on a machine is
its total load at that port. These workloads can be processed
independently, dropping the requirement that a transfer use its input
and output simultaneously~\cite{KhullerPurohit2016}.
For concurrent open shop without release dates, Mastrolilli et
al.~\cite{Mastrolilli2010} give a primal--dual $2$-approximation.
Concurrent open shop also embeds as a special case of coflow scheduling
using disjoint input--output port pairs. Thus the
$(2-\varepsilon)$-inapproximability of Sachdeva and
Saket~\cite{SachdevaSaket2013}, obtained through structural hardness for
hypergraph vertex cover, applies to coflows as well: for every fixed
$\varepsilon>0$, a $(2-\varepsilon)$-approximation would imply $P=NP$.
Qiu, Stein, and Zhong~\cite{Qiu2015} gave the first constant-factor
approximation algorithms using interval-indexed LP rounding,
geometric batching, and matching decompositions.
Khuller and Purohit~\cite{KhullerPurohit2016} instead obtain a coflow
ordering from a concurrent-open-shop relaxation and schedule geometric
batches in that order. They give deterministic factors $8$ and $12$
without and with release dates, respectively, and a randomized
$(3+2\sqrt{2})$-approximation without release dates.
Shafiee and Ghaderi~\cite{ShafieeGhaderi2017} combine an LP with ordering
variables and randomly shifted geometric batching to obtain randomized
factors $2e$ and $3e$, respectively.

Ahmadi et al.~\cite{Ahmadi2017} obtain deterministic factors $4$ and $5$
without and with release dates, respectively, using primal--dual coflow
ordering, moving edges into earlier batches
without increasing their maximum degrees, and decomposing the batches
into matchings. Shafiee and Ghaderi~\cite{ShafieeGhaderi2018} obtain the
same factors with a strengthened ordering-variable LP followed by
preemptive list scheduling.
Sincronia~\cite{Agarwal2018} connects approximation guarantees to system
design: primal--dual coflow ordering combined with work-conserving,
preemptive, order-preserving flow scheduling gives a factor $4$ when all
coflows arrive at time zero.
Fukunaga~\cite{Fukunaga2022} obtains factors $4$ and $5$ from a
time-indexed LP and greedy scheduling. Separately, a hypergraph-matching
argument gives a nonconstructive upper bound of $4$ on the LP's
integrality gap even with release dates.
Several earlier claims of $2$- or $(2+\varepsilon)$-approximations for
general bipartite coflow scheduling were subsequently found to rely on
incorrect arguments; see Ahmadi et al.~\cite{Ahmadi2017} and the authors'
error notice in Im and Purohit~\cite{ImPurohit2018}, as also discussed by
Rohwedder and Schnaars~\cite{RS2025}.
Im et al.~\cite{Im2019} give a $2$-approximation for
unit-length matroid coflows using LP-based deadline construction and
matroid intersection. 
Rohwedder and Schnaars~\cite{RS2025} improve the factors $4$ and $5$
for bipartite coflows to $140/41<3.42$ without release dates and $4.36$
with release dates. They construct fractionally feasible deadlines and
compare a greedy schedule with one obtained by Beck--Fiala-inspired
iterated rounding. The latter assigns edges to time blocks with
additive load violation at most two. The analysis exploits the
complementary behavior of these schedules for small and large
deadlines. 
The same work establishes valid asymptotic $(2+\varepsilon)$ guarantees
when the optimum is sufficiently large relative to the sum of coflow
weights~\cite{RS2025}.
Table~\ref{tab:comparison} summarizes the multiplicative approximation
guarantees for the general bipartite problem.
Rohwedder and Schnaars~\cite{RS2026} also study \emph{Graph Scheduling
with Group Completion Times}, which allows arbitrary multigraphs and
possibly overlapping edge groups. Degree and odd-set constraints,
iterated rounding, and multigraph edge coloring yield a bound of
$(2+\varepsilon)\operatorname{OPT}
+O(1/\varepsilon)\sum_j\omega_j$.
Comparing the rounding-based schedule with a greedy-type schedule gives
a $1104/281<3.93$-approximation. They also improve the guarantee for
unit-size Data Migration. Our focus here is on stronger multiplicative
guarantees for bipartite coflows, both without and with release dates.

\subsection{Our Contribution}\label{sec:contribution}

We give improved approximation algorithms for both variants of Coflow
Scheduling. Our algorithms use an iterated rounding procedure and,
for release dates, a scheduling construction based on stable matchings.

\begin{theorem}[Without release dates]\label{thm:main-no-release}
There is a deterministic polynomial-time algorithm achieving a
$58/21\ (<2.77)$-approximation for Coflow Scheduling without release dates.
\end{theorem}

\begin{theorem}[With release dates]\label{thm:main-release}
There is a deterministic polynomial-time algorithm achieving a
$1821/575\ (<3.17)$-approximation for Coflow Scheduling with nonnegative
integer release dates.
\end{theorem}

\begin{table}[htbp]
\caption{Approximation guarantees for Coflow Scheduling.
Decimal bounds are rounded upward to two places. D denotes a deterministic
algorithm and R an expected guarantee for a randomized algorithm.}
\label{tab:comparison}
\centering
\small
\setlength{\tabcolsep}{5pt}
\begin{tabular}{@{}lcrr@{}}
\toprule
Work & Type & No release dates & Release dates\\
\midrule
Qiu--Stein--Zhong (2015)~\cite{Qiu2015} & D & $\frac{64}{3}$ $(21.34)$ & $\frac{76}{3}$ $(25.34)^{\dagger}$\\
Qiu--Stein--Zhong (2015)~\cite{Qiu2015} & R & $8+\frac{16\sqrt{2}}{3}$ $(15.55)$ & ---\\
Khuller--Purohit (2016)~\cite{KhullerPurohit2016} & D & $8.00$ & $12.00$\\
Khuller--Purohit (2016)~\cite{KhullerPurohit2016} & R & $3+2\sqrt{2}$ $(5.83)$ & ---\\
Shafiee--Ghaderi (2017)~\cite{ShafieeGhaderi2017} & R & $2e$ $(5.44)$ &  $3e$ $(8.16)$\\
Ahmadi et al.\ (2017)~\cite{Ahmadi2017} & D & $4.00$ & $5.00$\\
Rohwedder--Schnaars ( 2025)~\cite{RS2025} & D & $140/41$ $(3.42)$ & $4.36$\\
\midrule
This paper & D & $58/21$ $\mathbf{(2.77)}$ &$1821/575$ $\mathbf{(3.17)}$\\
\bottomrule
\end{tabular}
\par\smallskip
\begin{minipage}{\linewidth}
$\dagger$ The release-date bound uses the correction of
Ahmadi et al.~\cite{Ahmadi2017} to the originally
claimed $67/3$ guarantee.
\end{minipage}
\end{table}

\paragraph{Iterated rounding} Our rounding procedure improves over the scheme of Rohwedder and Schnaars \cite{RS2025}, in which they used the result of Beck and Fiala~\cite{BeckFiala1981} to bound the error in terms of column sparsity of the matrix. Our method builds on ideas from iterative rounding techniques (see Lau, Ravi, and Singh~\cite{LRS2011}). Our algorithm selects a subset of tight capacity constraints and walks orthogonal to them to round the fractional solution. To perform this walk, we need to show that at every step the number of fractional variables is more than the number of selected rows. By doing a careful summation of capacity and assignment constraints we show that the number of fractional variables is at \textit{least} the number of selected rows. Then, we show that equality cannot take place, because it would imply linear dependence between the subset of capacity constraints that we have chosen and assignment constraints. Our intricate analysis yields two improvements. $(1)$ The rounding error for capacity constraints is bounded by $1$ instead of $2$ as in \cite{RS2025}. $(2)$ For one designated block, we can completely eliminate the error for capacity constraints in that block. Each block is a segment of timeslots, and by protecting one of the first blocks from the error, we greatly improve the approximation factor.


\paragraph{Release-aware scheduling through stable matchings}
For release dates, we use list-coloring algorithms after running our iterated rounding. Our idea to use list-coloring comes from the fact that for each edge $e$, we have a list of timeslots $L_e = \{r_e + 1, r_e + 2, \ldots, r_e + d_e\}$ in which it can be scheduled ($r_e$ is a release date of edge $e$, $d_e$ is how many slots we allow for this edge), and the list-coloring permits to assign an edge while preserving release and capacity constraints. The list-coloring method that we use is based on Galvin's method~\cite{Galvin1995}, which is implemented through
Gale--Shapley stable matchings~\cite{GaleShapley1962}. For the beautiful connection between stable matchings and list coloring,
see its application to the Dinitz problem in \emph{Proofs from THE
BOOK}~\cite[Chapter~38]{AignerZiegler2018} and the extensions of Galvin's
method by Borodin, Kostochka, and Woodall~\cite{BKW1997}. Moreover, we introduce an extension of greedy algorithm that achieves better guarantees than the standard greedy algorithm. Our extended algorithm uses a novel idea that we can reverse the priority of edges to balance out demand in a prefix of timeslots. This idea is of independent interest as it can be applied to other scheduling problems. Combined, our iterated rounding and the use of Galvin's method improve approximation factor for the setting with release dates.



\newpage
\subsection{Organization}

Section~\ref{sec:preliminaries} introduces the notation and recalls the required
background. Section~\ref{sec:rounding} introduces fractional schedules and
proves the iterated rounding theorem. Section~\ref{sec:no-release}
constructs the schedules without release dates, establishes their
completion-time bounds, and proves Theorem~\ref{thm:main-no-release}.
Section~\ref{sec:release} develops the release-date construction and its
polynomial implementation before proving Theorem~\ref{thm:main-release}.

\section{Definitions and Known Results}\label{sec:preliminaries}

We introduce the basic notation, formulate the Coflow Scheduling
problem, and recall the graph-theoretic results used in our proofs.

\subsection{Basic graph notions}
 
A \emph{multigraph} is a triple $G=(V,E,\partial)$, where $V$ is
the vertex set, $E$ is the edge set, and
$\partial:E\to\binom{V}{2}$ assigns to each edge its unordered pair
of endpoints. Distinct edges may have the same endpoints; such
edges are called \emph{parallel}. If $\partial$ is injective,
$G$ is a \emph{simple graph}.

A multigraph is \emph{bipartite} if its vertex set admits a partition
$V=V_1\sqcup V_2$ such that every edge has one endpoint in $V_1$ and
the other in $V_2$.

For each vertex $v\in V$, let
$\delta(v)=\{e\in E:v\in\partial(e)\}$ denote the set of edges
\emph{incident} with $v$. For $F\subseteq E$, define
\[
 \deg_F(v)=|F\cap\delta(v)|,
 \qquad
 \Delta(F)=\max_{v\in V}\deg_F(v).
\]
We also write $\Delta(G)=\Delta(E)$ for the maximum degree of $G$.
A set $M\subseteq E$ is a \emph{matching} if $\Delta(M)\le1$.

\subsection{Coflow scheduling problem}

In this section, we introduce the coflow scheduling problem. An instance consists of $n$ bipartite multigraphs
$G_j=(V,E_j,\partial_j)$ on a common bipartition $V=V_1\sqcup V_2$,
together with positive weights $\omega_j$ and release dates
$r_j\in\mathbb{Z}_{\ge0}$, where $j\in[n]:=\{1,\ldots,n\}$.
Each $G_j$ is called a \emph{coflow}, which we identify with its
edge set $E_j$. We assume that all coflows are nonempty.

Let
\[
 E=\bigsqcup_{j\in[n]}E_j,
 \qquad G=(V,E,\partial),
\]
where $\partial|_{E_j}=\partial_j$. All edge copies are treated as
distinct, including parallel copies within the same coflow or in
different coflows. For $e\in E_j$, write $j(e)=j$ and $r_e=r_j$. Later, we will introduce coflow deadlines $C_1, \ldots, C_n$. The deadlines on coflows induce the same deadlines on edges $C_e=C_{j(e)}$.
Sometimes we abuse the notation and write $e = \{u, v\}$ instead of $\partial(e) = \{u, v\}$.

Each edge requires one time slot of processing, and slot $t$
corresponds to the interval $(t-1,t]$. A \emph{feasible schedule}
is a map $\sigma:E\to\mathbb{Z}_{\ge1}$ such that
$\sigma^{-1}(t)$ is a matching for every slot $t$ and $ \sigma(e)\ge r_{j(e)}+1$ for each $e \in E$.
Thus each vertex participates in at most one scheduled edge per
slot, and no edge is processed before its coflow's release date.

The \emph{completion time} of coflow $E_j$ and the \emph{cost} of
a schedule $\sigma$ are
\[
 C_j^*(\sigma)=\max_{e\in E_j}\sigma(e),
 \qquad
 \cost(\sigma)=\sum_{j\in[n]}\omega_j C_j^*(\sigma).
\]

\begin{definition}[Coflow Scheduling with Release Dates]
Given coflows $G_1,\ldots,G_n$, positive weights
$\omega_1,\ldots,\omega_n$, and nonnegative integer release dates
$r_1,\ldots,r_n$, find a feasible schedule $\sigma$ of minimum $\mathrm{cost}(\sigma)$.
\end{definition}

We denote the minimum cost by $\OPT$. The special case $r_j=0$
for every $j\in[n]$ is called \emph{Coflow Scheduling without
Release Dates}.

To formally illustrate the problem, let us reformulate Coflow Scheduling with Release Dates as an integer
linear program. Note that we will not use this program in our proofs. Let
\[
 R=\max_{j\in[n]}r_j,
 \qquad T=R+|E|.
\]
We may restrict attention to schedules that finish by $T$, because each edge can be scheduled in its own slot.

For each $e\in E$ and $t\in[T]$, introduce a binary variable
$x_{e,t}$ indicating whether edge $e$ is scheduled in slot $t$.
For each coflow $E_j$, introduce an integer variable $C_j$ bounding
its completion time from above. The formulation is
\begingroup
\renewcommand{\theequation}{ILP}
\begin{subequations}\label{eq:coflow-ilp}
\renewcommand{\theequation}{ILP.\arabic{equation}}
\begin{align}
 \min\quad
 &\sum_{j\in[n]}\omega_j C_j
 &&\label{eq:ILP-objective}\\
 \text{subject to}\quad
 &\sum_{t\in[T]}x_{e,t}=1
 &&\forall e\in E,
 \label{eq:ILP-assignment}\\
 &\sum_{e\in\delta(v)}x_{e,t}\le1
 &&\forall v\in V,\ \forall t\in[T],
 \label{eq:ILP-matching}\\
 &x_{e,t}=0
 &&\forall e\in E,\ \forall t\in[T]:t \leq r_e,
 \label{eq:ILP-release}\\
 &\sum_{t\in[T]}t\,x_{e,t}\le C_{j(e)}
 &&\forall e\in E,
 \label{eq:ILP-completion}\\
 &x_{e,t}\in\{0,1\}
 &&\forall e\in E,\ \forall t\in[T],
 \label{eq:ILP-binary}\\
 &C_j\in\{r_j+1,\ldots,T\}
 &&\forall j\in[n].
 \label{eq:ILP-integer}
\end{align}
\end{subequations}
\endgroup

Constraints~\eqref{eq:ILP-assignment} schedule every edge exactly
once, and \eqref{eq:ILP-matching} ensure that the edges scheduled
in each slot form a matching.
Constraints~\eqref{eq:ILP-release} forbid processing before the
coflow's release date, while \eqref{eq:ILP-completion} require
$C_j$ to be at least the slot assigned to every edge of $E_j$.
Constraints~\eqref{eq:ILP-binary} and~\eqref{eq:ILP-integer}
specify the variable domains.

Every feasible solution defines a feasible schedule by setting
$\sigma(e)=t$ whenever $x_{e,t}=1$. Since all weights are positive,
every optimal solution satisfies
\[
 C_j=\max_{e\in E_j}\sum_{t\in[T]}t\,x_{e,t}
     =C_j^*(\sigma)
 \qquad(j\in[n]).
\]

When all release dates are zero, $T=|E|$ and
\eqref{eq:ILP-release} is vacuous, recovering the formulation
without release dates. 

\subsection{Known graph results}

In this section, we provide necessary background definitions and results about graphs and matchings that we will use in our algorithms. We say an edge coloring is \textit{proper} if two different edges sharing a vertex have different colors. Consider timeslots $1, 2, 3, \ldots$ as different colors. In this view, a proper edge coloring partitions the edges into matchings and
therefore gives a schedule when release constraints are absent (color corresponds to timeslot in which the edge is scheduled).
We use the following classical theorem.

\begin{lemma}[{K\H{o}nig~\cite{Konig1916}}]
\label{lem:konig}
Every bipartite multigraph $F$ admits a proper edge coloring with
$\Delta(F)$ colors, and such a coloring can be found in polynomial time.
Equivalently, its edges can be scheduled in $\Delta(F)$ slots.
\end{lemma}

Our algorithm for instances with release dates uses stable matchings
and the Gale--Shapley theorem. We recall the relevant definitions
and theorem below.
For each vertex $v\in V$, fix a strict total order $\succ_v$ on
its incident edges $\delta(v)$, where $e\succ_v f$ means that $v$
prefers $e$ to $f$. A matching $F\subseteq E$ is \emph{stable} if, for every edge
$e\in E\setminus F$, there exist an endpoint $v$ of $e$
and an edge $f\in F\cap\delta(v)$ such that $f\succ_v e$.
\begin{theorem}[Gale--Shapley \cite{GaleShapley1962}]
\label{thm:gale-shapley}
Every finite bipartite multigraph equipped with strict total
orders on the incident edges at each vertex admits a stable
matching. Such a matching can be found in polynomial time.
\end{theorem}

\section{Fractional Scheduling and Iterated Rounding}\label{sec:rounding}

In this section, we prove our main iterated rounding result. 
Given an instance of Coflow Scheduling without release dates, with
coflows $E_1,\ldots,E_n$ and weights $\omega_1,\ldots,\omega_n$, we first
use Lemma~\ref{lem:deadlines} to compute fractionally feasible deadlines
$C_1,\ldots,C_n$. We then enlarge these deadlines to obtain integer-length
blocks and apply iterated rounding to assign each edge to one allowed
block. The rounding guarantees will also be used for instances with
release dates.

\subsection{Fractional Schedule}\label{sec:deadlines-overview}
Let $C_1,\ldots,C_n$ be deadlines for coflows $E_1,\ldots,E_n$,
respectively. We assume that the coflows are relabeled so that
$C_1\le\cdots\le C_n$, and set $C_0=0$.
We say that the deadlines are \emph{fractionally feasible} if the
following linear program admits a solution:
\begingroup
\renewcommand{\theequation}{LP I}
\begin{subequations}\label{lp:I}
\renewcommand{\theequation}{LP I.\arabic{equation}}
\begin{align}
 \sum_{s\in[n]}x_{s,e}&=1
       &&\forall e\in E,
       \label{eq:lp-I-assignment}\\
 \sum_{e\in\delta(v)}x_{s,e}&\le C_s-C_{s-1}
       &&\forall s\in[n],\ \forall v\in V,
       \label{eq:lp-I-capacity}\\
 x_{s,e}&=0
       &&\forall j\in[n],\ \forall e\in E_j,\ \forall s>j,
       \label{eq:lp-I-deadline}\\
 x_{s,e}&\ge0
       &&\forall s\in[n],\ \forall e\in E.
       \label{eq:lp-I-nonnegative}
\end{align}
\end{subequations}
\endgroup
Here, the timeslots are partitioned into blocks $(0, C_1] \sqcup (C_1, C_2] \sqcup \ldots \sqcup (C_{n-1}, C_n]$, and $x_{s,e}$ is the fraction of edge $e$ assigned to the block
$(C_{s-1},C_s]$. Constraints~\eqref{eq:lp-I-assignment} ensure that every edge of each coflow is fully scheduled, and
\eqref{eq:lp-I-capacity} ensure that the fractional degree of each vertex in every block does not exceed the block size $C_{s}-C_{s-1}$. Constraints~\eqref{eq:lp-I-deadline} forbid edge scheduling after the coflow's deadline.

To construct fractionally feasible deadlines, we use the result of \cite[Section~4]{Im2019}. We use its deterministic,
real-deadline form described by Rohwedder and
Schnaars~\cite[Lemma~5 and Section~3.1]{RS2025}.

\begin{lemma}[{Im--Moseley--Pruhs--Purohit~\cite{Im2019}}]
\label{lem:deadlines}
In deterministic polynomial time one can compute deadlines
$1\le C_1\le\cdots\le C_n$, after relabeling the coflows, for which
\eqref{lp:I} is feasible and
\begin{equation}\label{eq:deadline-potential}
 \sum_{j\in[n]}\omega_j(C_j+1)\le2\OPT.
\end{equation}
\end{lemma}

The deadlines provided by Lemma~\ref{lem:deadlines} need not be integer. We will round them upward. The following lemma shows that this preserves fractional feasibility.

\begin{lemma}[Deadline monotonicity]\label{lem:deadline-monotonicity}
If \eqref{lp:I} is feasible for $C_1\le\cdots\le C_n$, then it remains
feasible after replacing $C_j$ by $\overline C_j$ for every $j$,
provided that $\overline C_1\le\cdots\le\overline C_n$ and
$\overline C_j\ge C_j$ for all $j$.
\end{lemma}

\begin{proof}
Let $x$ be feasible for the original deadlines $C_j$. Let us smear edges over the timeline. More precisely, on each positive-length
block $(C_{s-1},C_s]$, process edge $e$ at the constant rate
$x_{s,e}/(C_s-C_{s-1})$, with no processing outside $(0,C_n]$. The constraints of \eqref{lp:I}
ensure one unit of processing per edge, total rate at most one at each
vertex at any given time, and no processing of $E_j$ after $C_j$.

Set $\overline C_0=0$, and let $\overline x_{s,e}$ be the amount of this processing in
$(\overline C_{s-1},\overline C_s]$. Since $\overline C_n\ge C_n$,
these variables are nonnegative and satisfy
\[
 \sum_{s\in[n]}\overline x_{s,e}=1,
 \qquad
 \sum_{e\in\delta(v)}\overline x_{s,e}
 \le\overline C_s-\overline C_{s-1}.
\]
For $e\in E_j$ and $s>j$, we have
$\overline C_{s-1}\ge\overline C_j\ge C_j$, so
$\overline x_{s,e}=0$. Hence $\overline x$ satisfies \eqref{lp:I}
with the enlarged deadlines.
\end{proof}

\subsection{The Iterated Rounding Procedure}\label{sec:iterated-rounding}

We now round a feasible solution of \eqref{lp:I} for the enlarged
deadlines $\overline C_j$. The procedure preserves the assignment and
deadline constraints and allows an additive violation of at most one
in each capacity constraint. In any designated block $p$, it preserves
all capacities exactly. Thus the load at each vertex is at most
$\overline C_s-\overline C_{s-1}+1$ in every block $s\ne p$, and at most
$\overline C_p-\overline C_{p-1}$ in block $p$.
For the algorithm without release dates, we choose $p=1$.

We state the rounding result in a general form. Let $\calB$ be a finite
set of blocks, and let each edge $e\in E$ have a nonempty subset
$\calB_e\subseteq\calB$ of blocks in which it can be scheduled.
A set $\mathcal{C}\subseteq V\times\calB$ specifies the vertex--block
pairs with nonnegative integer capacities $c_{v,b}$.
One designated block $p\in\calB$ is called the \emph{protected block};
all other blocks are \emph{ordinary}.
For \eqref{lp:I} with the enlarged deadlines, we take $\calB=[n]$,
$\calB_e=\{1,\ldots,j(e)\}$, $\mathcal{C}=V\times[n]$, and
$c_{v,b}=\overline C_b-\overline C_{b-1}$.

\begin{theorem}[{Iterated rounding}]\label{lem:rounding}

Let $G=(V_1\mathbin{\sqcup}V_2,E)$ be a bipartite
multigraph, and let $\mathcal{B}$ denote a set of blocks. Let $\calB_e\subseteq\calB$ be nonempty
for every $e\in E$, and let
$\mathcal C\subseteq V\times\calB$ be a subset that indexes constraints (\ref{eq:ir-capacity}), with
$c_{v,b}\in\mathbb Z_{\ge0}$ for $(v,b)\in\mathcal C$.
Given a fractional solution $z_{b,e} \in [0, 1]$ of the system

\begin{align}
 \sum_{b\in\calB_e}z_{b,e}&=1 &&(e\in E),\label{eq:ir-assign}\\
 \sum_{e\in\delta(v):\,b\in\calB_e}z_{b,e}&\le c_{v,b}
       &&((v, b) \in \mathcal{C}),\label{eq:ir-capacity}
\end{align}

and any
$p\in\calB$, one can compute in polynomial time
$x_{b,e}\in\{0,1\}$ satisfying
\begin{align}
\sum_{b\in\calB_e}x_{b,e}=1 \quad(e\in E),
\qquad
\sum_{e\in\delta(v):\,b\in\calB_e}x_{b,e}
\le c_{v,b}+\mathbf1_{\{b\ne p\}}
\quad((v,b)\in\mathcal C).\label{eq:ir-guarantee}
\end{align}
\end{theorem}

\begin{proof}
   We refer to constraints (\ref{eq:ir-assign}) as assignment constraints and to constraints (\ref{eq:ir-capacity}) as capacity constraints. Initialize $x_{b, e} \leftarrow z_{b, e}$. Let $\mathcal{F}$ denote the set of current variables $(x_{b, e})$, and let $\mathcal{F}^{(v, b)} = \{x_{b, e} \in \mathcal{F} ~|~ e \in E, e \ni v \}$. For $\mathcal A\subseteq\mathcal F$, write
$S(\mathcal A)=\sum_{x_{b,e}\in\mathcal A}x_{b,e}$. Note that with this notation a capacity constraint $(v, b) \in \mathcal{C}$ can be written as $S(\mathcal{F}^{(v, b)}) \leq c_{v, b}$.

   The algorithm repeatedly applies one of the following three rules.
We show that at least one rule is available whenever a coordinate remains.

    \begin{enumerate}
        \item Look for variables $x_{b, e}$ that are already rounded $x_{b, e} \in \{0, 1\}$. We drop such variables from our system $\mathcal{F} \leftarrow \mathcal{F} \setminus x_{b, e}$, simplifying it. We also remove $b$ from $\calB_e$. If $x_{b, e} = 1$, we also decrease the corresponding integers $c_{v, b}$ by $1$ and remove $E \leftarrow E \setminus \{e\}$ since the edge $e$ is assigned to block $b$. Therefore, w.l.o.g. we may assume that $x_{b, e} \in (0, 1)$ for every variable in our system.

        \item Look for capacity constraints $S(\mathcal{F}^{(v, b)}) \le c_{v,b}$ ($(v, b) \in \mathcal{C}$), such that $|\mathcal{F}^{(v, b)}| \leq c_{v, b} + 1$ if $p \neq b$ and $|\mathcal{F}^{(v, b)}| \leq c_{v, b}$ if $p = b$. In this case we can drop this constraint from consideration ($\mathcal{C} \leftarrow \mathcal{C} \setminus (v, b)$) since any rounding of the remaining variables will violate it only by $1$ if $p \neq b$ and will not violate it at all if $p = b$. 

        \item Consider the set $\mathcal{C}^{\textit{tight}} \subseteq \mathcal{C}$ of capacity constraints that are tight: $S(\mathcal{F}^{(v, b)}) = c_{v,b}$. Let $\mathcal{C}' \subseteq \mathcal{C}^{\textit{tight}}$ be a subset of tight constraints that together with all hard constraints $\sum_{b\in\calB_e}x_{b,e} = 1$ forms a linearly independent basis. Here we view
        a constraint as a vector from $\{0, 1\}^{\mathcal{F}}$ ($1$ if a corresponding variable is in a constraint and $0$ otherwise). If $|\mathcal{F}| > |E| + \mathcal{C}'$, then there is a non-zero vector $y \in \mathbb{R}^{\mathcal{F}}$ that is orthogonal to all constraints from our basis: $\sum_{e\in\delta(v):\,b\in\calB_e}y_{b,e} = 0$ and $\sum_{b \in \calB_e} y_{b, e} = 0$. We slide along this vector until some variable gets rounded or some other constraint becomes tight: $x \leftarrow x + ty$. 

    \end{enumerate}
    See Algorithm (\ref{alg:iterated-rounding}) for the pseudocode. Note that rules $1$ and $2$ reduce the number of variables or constraints at least by $1$. The $3$rd rule either increases the number of rounded variables or the number of tight constraints. Therefore, our algorithm makes at most polynomially many steps before terminating (it terminates when $|\mathcal{F}| = 0$). It remains to show that at least one of three rules is always applicable. The rest of the proof is devoted to this.
    
    \paragraph{\textbf{Warm-up: $+1$ error for all constraints.} } Firstly, suppose that we do not need to protect block $p$ (e.g. $p = -1$). Let us show that in this case at least one rule can be applied (we also assume that there is at least one variable in our system). Towards contradiction, assume that no rule is applicable. Rules $2$ and $3$ imply that for any $(v, b) \in \mathcal{C}'$, we have $|\mathcal{F}^{(v, b)}| \geq c_{v,b} + 2$ since $|\mathcal{F}^{(v, b)}|$ and $c_{v, b}$ are integers. Since $(v, b) \in \mathcal{C}^{\textit{tight}}$, we also have $S(\mathcal{F}^{(v, b)}) = c_{v,b}$. Substitution gives

  \begin{equation} \label{eq:b1}
    |\mathcal{F}^{(v, b)}| \geq S(\mathcal{F}^{(v, b)}) + 2.
\end{equation}
Note that if we sum assignment constraints $\sum_{b \in \calB_e} x_{b, e} = 1$ over all $e \in E$, we get that $S(\mathcal{F}) = |E|$. Our goal is to show that $|\mathcal{F}| > |E| + |\mathcal{C}'|$, which would produce the desired contradiction since the rule $3$ can be applied. Note that (\ref{eq:b1}) states that "the number of variables in the constraint is at least their sum $+2$", while we also know that the sum of all variables is $|E|$. Let us sum inequalities (\ref{eq:b1}) over all $(v, b) \in \mathcal{C}'$. For $t \in \{0, 1, 2\}$, denote by $\mathcal{F}_t \subseteq \mathcal{F}$ variables that are summed $t$ times ($\mathcal{F}_0 \sqcup \mathcal{F}_1 \sqcup \mathcal{F}_2 = \mathcal{F}$ holds since every edge is incident to at most two vertices). With this notation, the summation gives

\begin{equation} \label{eq:b2}
|\mathcal{F}_1| + 2|\mathcal{F}_2|  \geq S(\mathcal{F}_1) + 2S(\mathcal{F}_2) + 2|\mathcal{C}'|.
\end{equation}

Let us show that $|\mathcal{F}_0| + |\mathcal{F}_1| > 0$. If $|\mathcal{F}_0| = |\mathcal{F}_1| = 0$, then $\mathcal{C}' = \mathcal{C}$ since only by including every constraint in the basis can we make each variable appear twice in the sum. However, in this case capacity constraints $\mathcal{C}$ and assignment constraints $E$ are not linearly independent. The sum of the first ones equals twice the sum of second ones, implying that $\mathcal{C}'$ cannot equal $\mathcal{C}$. Thus, $|\mathcal{F}_0| + |\mathcal{F}_1| > 0$. Since every variable is strictly between $0$ and $1$, we also get that $2|\mathcal{F}_0| + |\mathcal{F}_1| > 2S(\mathcal{F}_0) + S(\mathcal{F}_1)$.

Adding this inequality to (\ref{eq:b2}) we obtain $|\mathcal{F}| > S(\mathcal{F})  +  |\mathcal{C}'|$. Substituting $S(\mathcal{F}) = |E|$, implies $|\mathcal{F}| > |E| + |\mathcal{C}'|$ as desired by the $3$rd rule.

\paragraph{\textbf{Protecting block $p$.}} Now, let us modify the proof above to handle the case $p \in \mathcal{B}$. Towards contradiction, assume that no rule is applicable. Due to the $2$nd rule, for any $(v, b) \in \mathcal{C}'$ we have
 \begin{equation} \label{eq:b3}
    |\mathcal{F}^{(v, b)}| \geq S(\mathcal{F}^{(v, b)}) + 2 - \ind_{\{b=p\}}.
\end{equation} 

We are going to do a more careful summation than we did in the warm-up. For this, for a vertex $v \in V$, define $\mathcal{F}^v = \{ x_{b, e} \in \mathcal{F} ~|~ e \ni v\}$ and $\mathcal{C}'_v = \{(v, b) \in \mathcal{C'} ~|~ b \in \calB\}$ to be the set of variables and constraints adjacent to vertex $v$. We show the following claim.

\begin{claim}
    For any vertex $v \in V$, $|\mathcal{F}^v| \geq S(\mathcal{F}^v) + 2 |\mathcal{C}'_v|$. Moreover, for at least one $v \in V$ the inequality is strict (assuming $|\mathcal{F}| > 0$).
\end{claim}

First, let us show why this inequality implies contradiction. Summing the inequality over $v \in V$ every variable gets summed exactly twice since it is adjacent to exactly two vertices. Therefore, we have $2 |\mathcal{F}| > 2S(\mathcal{F}) + 2|\mathcal{C}'|$, implying $|\mathcal{F}| > |E| + |\mathcal{C}'|$, contradicting the assumption that Rule~3 is unavailable..

It remains to prove the claim. Fix vertex $v \in V$, and sum inequalities (\ref{eq:b3}) over $(v, b) \in \mathcal{C}'_v$. Denoting by $\mathcal{F}^v_t \subseteq \mathcal{F}^v$ the subset of variables that appears $t$ times in the sum, we get

\begin{equation} \label{eq:b4}
|\mathcal{F}^v_1| \geq S(\mathcal{F}_1^v) + 2|\mathcal{C}'_v| - I_{\{(v, p) \in \mathcal{C}'_v\}}.
\end{equation}
Assume $|\mathcal{F}^v| > 0$ (otherwise the first part of the claim is trivial). Let us establish that $|\mathcal{F}_0^v| > 0$ and that $S(\mathcal{F}_0^v)$ is integer, which would imply that $|\mathcal{F}_0^v| \geq S(\mathcal{F}_0^v) + 1$. Adding this inequality to (\ref{eq:b4}) would prove the first part of the claim (note that $\mathcal{F}^v = \mathcal{F}^v_0 \sqcup \mathcal{F}^v_1$). If $|\mathcal{F}^v_0| = 0$, then $\mathcal{C}'_v$ must include all capacity constraints adjacent to vertex $v$. However, in this case constraints $\mathcal{C}'_v$ are linearly dependent with assignment constraints of edges $E^v$ (edges adjacent to $v$), contradicting definition of $\mathcal{C}'$ (see Figure~\ref{fig:vlindep} for illustration). Therefore, $|\mathcal{F}^v_0| > 0$. By summing assignment constraints of edges $E^v$, we get $S(\mathcal{F}_0^v) + S(\mathcal{F}_1^v) = |E^v| \in \mathbb{Z}$. Note that $S(\mathcal{F}_1^v) = \sum_{(v, b) \in \mathcal{C}'_v} c_{v, b} \in \mathbb{Z}$, since $c_{v, b} \in \mathbb{Z}$, implying that $S(\mathcal{F}_0^v)$ is integer as well. Thus, the first part of the claim is established.

Let us establish the second part of the claim. Note that if for at least one vertex $v \in V$, such that $|\mathcal{F}^v| > 0$, we have $(v, p) \notin \mathcal{C}'_v$, then $|\mathcal{F}^v| \geq S(\mathcal{F}^v) + 2|\mathcal{C}'_v| + 1$. Towards contradiction, assume that for any $v \in V$, such that $|\mathcal{F}^v| > 0$, we have $(v, p) \in \mathcal{C}'_v$. Then, let us show that constraints from $\mathcal{C}'$ are not linearly independent. Indeed, the sets of constraints $\{ (v, p) ~|~ v \in V_1, |\mathcal{F}^v| > 0\}$  and $\{ (v, p) ~|~ v \in V_2, |\mathcal{F}^v| > 0\}$ touch every variable $x_{p, e}$ exactly once. Therefore, the sum of the coefficient vectors of the first ones equals the sum of the second ones, exhibiting linear dependence. Thus, there is at least one vertex $v$ for which the inequality is strict, finishing the proof of the claim.
\end{proof}

\begin{figure}[htbp]
    \centering

\begin{tikzpicture}[
    mynode/.style={draw, circle, thick, minimum size=0.55cm, inner sep=0pt},
    scale=0.79,
    transform shape
]

\def\drawnodes{
    \foreach \y in {0,1,2} {
        \node[mynode] (L\y) at (0,\y) {};
        \node[mynode] (R\y) at (2,\y) {};
    }
}


\begin{scope}[shift={(0,0)}]
    \drawnodes
    \node at (1, 3.2) {\large \textbf{Block 1}};
    \node[left] at (-0.35, 1) {\Large $v$};

    \draw[thick, blue] (L1) -- (R2) node[pos=0.4, above left=-2pt, color=black] {$x_{1,e}$};
    \draw[thick, blue] (L1) -- (R1) node[pos=0.7, above=1pt, color=black] {$x_{1,e'}$};
    \draw[thick, blue] (L1) -- (R0) node[pos=0.7, above right=-2pt, color=black] {$x_{1,e''}$};
\end{scope}

\begin{scope}[shift={(4.5,0)}]
    \drawnodes
    \node at (1, 3.2) {\large \textbf{Block 2}};

    \draw[thick, green!60!black] (L1) -- (R2) node[pos=0.6, below right=-2pt, color=black] {$x_{2,e}$};
    \draw[thick, green!60!black] (L1) -- (R0) node[pos=0.7, above right=-2pt, color=black] {$x_{2,e''}$};
\end{scope}

\begin{scope}[shift={(9,0)}]
    \drawnodes
    \node at (1, 3.2) {\large \textbf{Block 3}};

    \draw[thick, red] (L1) -- (R1) node[pos=0.5, above=1pt, color=black] {$x_{3,e'}$};
    \draw[thick, red] (L1) -- (R0) node[pos=0.7, above right=-2pt, color=black] {$x_{3,e''}$};
\end{scope}

\begin{scope}[shift={(5.5,-1.0)}]
    \draw[ultra thick] (-0.15, -0.45) -- (-0.15, 0.45);
    \draw[ultra thick]  (0.15, -0.45) --  (0.15, 0.45);
\end{scope}


\begin{scope}[shift={(0,-5.5)}]
    \drawnodes
    \node at (1, 3.2) {\large \textbf{Block 1}};
    \node[left] at (-0.35, 1) {\Large $v$};

    \draw[thick, blue] (L1) -- (R2) node[pos=0.4, above left=-2pt, color=black] {$x_{1,e}$};
    \draw[thick, green!60!black] (L1) -- (R1) node[pos=0.8, above=1pt, color=black] {$x_{1,e'}$};
    \draw[thick, red] (L1) -- (R0) node[pos=0.7, above right=-2pt, color=black] {$x_{1,e''}$};
\end{scope}

\begin{scope}[shift={(4.5,-5.5)}]
    \drawnodes
    \node at (1, 3.2) {\large \textbf{Block 2}};

    \draw[thick, blue] (L1) -- (R2) node[pos=0.6, below right=-2pt, color=black] {$x_{2,e}$};
    \draw[thick, red] (L1) -- (R0) node[pos=0.7, above right=-2pt, color=black] {$x_{2,e''}$};
\end{scope}

\begin{scope}[shift={(9,-5.5)}]
    \drawnodes
    \node at (1, 3.2) {\large \textbf{Block 3}};

    \draw[thick, green!60!black] (L1) -- (R1) node[pos=0.5, above=1pt, color=black] {$x_{3,e'}$};
    \draw[thick, red] (L1) -- (R0) node[pos=0.7, above right=-2pt, color=black] {$x_{3,e''}$};
\end{scope}

\end{tikzpicture}

\caption{Illustration of the linear dependence between the assignment constraints and the capacity constraints adjacent to vertex $v$. In the top row, we have capacity constraints $\textcolor{blue}{x_{1, e} + x_{1, e'} + x_{1, e''} = c_{v, 1}}$, $\textcolor{OliveGreen}{x_{2, e} + x_{2, e''} = c_{v, 2}}$, and $\textcolor{red}{x_{3, e'} + x_{3, e''} = c_{v, 3}}$. In the bottom row, we have assignment constraints $\textcolor{blue}{x_{1, e} + x_{2, e} = 1}$, $\textcolor{OliveGreen}{x_{1, e'} + x_{3, e'} = 1}$, $\textcolor{red}{x_{1, e''} + x_{2, e''} + x_{3, e''} = 1}$. Each variable occurs exactly once in capacity constraints and exactly once in assignment constraints, implying that the two families of constraints are linearly dependent.} 
\label{fig:vlindep}

\end{figure}

\begin{algorithm}[H]
\caption{Iterated rounding procedure}
\label{alg:iterated-rounding}
\small
\begin{algorithmic}[1]
\Require A fractional solution $z_{b, e} \in [0, 1]$ to the system~\eqref{eq:ir-assign}--\eqref{eq:ir-capacity}
with integer capacities and a designated block $p$. Assignment and capacity constraints are indexed by edges $E$ and a subset $\mathcal{C} \subseteq V \times \mathcal{B}$ of vertex-block pairs, respectively.
\Ensure An integer solution $x_{b, e} \in \{0, 1\}$ satisfying~\eqref{eq:ir-guarantee}.
\State Initialize $x_{b, e} \leftarrow z_{b, e}$
\State $\mathcal{F} \gets \{x_{b,e} \mid e \in E, b \in \calB_e\}$ \Comment{Set of current variables}
\While{$\mathcal{F} \neq \varnothing$}
  \If{there is $x_{b,e} \in \mathcal{F}$ with $x_{b,e} \in \{0, 1\}$}
    \State $\mathcal{F} \gets \mathcal{F} \setminus \{x_{b,e}\}$ \Comment{Rule 1: Drop rounded variables}
    \State $\calB_e \gets \calB_e \setminus \{b\}$
    \If{$x_{b,e} = 1$}
      \For{$v \in e$ with $(v, b) \in \mathcal{C}$}
        \State $c_{v,b} \gets c_{v,b} - 1$ 
      \EndFor
      \State $E \leftarrow E \setminus \{e\}$
    \EndIf
  \ElsIf{there is $(v,b) \in \mathcal{C}$ s.t. $|\mathcal{F}^{(v,b)}| \le c_{v,b} + 1$ (if $p \ne b$) or $|\mathcal{F}^{(v,b)}| \le c_{v,b}$ (if $p = b$)}
    \State $\mathcal{C} \gets \mathcal{C} \setminus \{(v,b)\}$ \Comment{Rule 2: Drop relaxed constraints}
  \Else
    \State  $\mathcal{C}^{\textit{tight}} \gets \{(v,b) \in \mathcal{C} \mid S(\mathcal{F}^{(v,b)}) = c_{v,b}\}$
    \State Define $a^{(e)} \in \{0, 1\}^{\mathcal{F}}$ as indicator of variables present in the assignment constraint $e \in E$
    \State Define $a^{(v, b)} \in \{0, 1\}^{\mathcal{F}}$ as indicator of variables present in the capacity constraint $(v, b) \in \mathcal{C}$
    \State Find $0 \not=y \in \mathbb{R}^{\mathcal{F}}$ s.t. $\langle a^{(e)}, y \rangle = 0$ for all $e \in E$ and $\langle a^{(v,b)}, y \rangle = 0$ for all $(v,b) \in \mathcal{C}^{\textit{tight}}$
    \State Choose $t \in \mathbb{R}$ such that $x' = x + ty \in [0, 1]^{\mathcal{F}}$ and $\langle a^{(v,b)}, x' \rangle \le c_{v,b}$ for all $(v,b) \in \mathcal{C}$, with the property that either $x'_{b,e} \in \{0, 1\}$ for some variable or $\langle a^{(v,b)}, x' \rangle = c_{v,b}$ for some $(v,b) \notin \mathcal{C}^{\textit{tight}}$
    \State $x \gets x + ty$ \Comment{Rule 3: Slide along kernel}
  \EndIf
\EndWhile
\end{algorithmic}
\end{algorithm}

\section{Coflow Scheduling without Release Dates}\label{sec:no-release}

Fix deadlines $C_1\le\cdots\le C_n$ for which \eqref{lp:I} is feasible.
We first construct schedules by iterated rounding and then establish
their approximation guarantee.

\subsection{Scheduling via Iterated Rounding}\label{sec:calendar-proof}

For any positive integers $\lambda,\tau$, we define a scheduling algorithm $\IRA^\tau_\lambda$, which takes as input deadlines $C_1,\ldots,C_n$ for which \eqref{lp:I} is feasible and constructs a valid coflow schedule in three steps:

\textbf{1. Deadline Rounding} 

The algorithm rounds each deadline $C_j$ up to the smallest element
of $\{\lambda+i\tau:i\in\mathbb{Z}_{\ge0}\}$ that is at least $C_j$.
Denoting the rounded deadline by $\overline C_j$, we have
\begin{equation}\label{eq:calendar}
 \overline C_j=\lambda+\tau\max\left\{0,
                  \ceil{\frac{C_j-\lambda}{\tau}}\right\}.
\end{equation}

\textbf{2. Block Formation and Iterated Rounding}

Rounded deadlines may coincide. Let $0 < T_1<\cdots<T_B$ be the distinct values of the rounded deadlines $\overline{C}_1, \dots, \overline{C}_n$. In other words, $\{T_1, \ldots, T_B\} = \{\overline{C}_1, \ldots, \overline{C}_n\}$. We set $T_0=0$, and form blocks $(T_{b-1},T_b]$ for $b\in [B]$.
Their lengths are integers, with $T_1\ge\lambda$ and
$T_b-T_{b-1}\ge\tau$ for $b\ge2$.
Since $\overline C_j\ge C_j$ for all $j$,
Lemma~\ref{lem:deadline-monotonicity} ensures that \eqref{lp:I} remains
feasible for the rounded deadlines.

    We apply Theorem~\ref{lem:rounding} with
\[
\mathcal B=[B],\qquad 
\mathcal B_e=\{b\in[B]:T_b\le \overline C_{j(e)}\},\qquad
\mathcal C=V\times[B],
\]
and
\[
c_{v,b}=T_b-T_{b-1}.
\]
Take a feasible solution of \eqref{lp:I} for the rounded deadlines,
discard its zero-length blocks, and index the remaining blocks by
$b\in[B]$. Let $z_{b,e}$ be the resulting assignment.
Then $z$ satisfies \eqref{eq:ir-assign}--\eqref{eq:ir-capacity}
with the parameters above. Applying Theorem~\ref{lem:rounding} with protected block \(p=1\) assigns every edge to one block \(b\in\mathcal B_e\), while violating each vertex--block capacity by at most \(1\) for \(b\ne1\) and not violating it for \(b=1\).

\textbf{3. Scheduling via K\H{o}nig's Theorem}

Finally, the algorithm applies K\H{o}nig's theorem
(Lemma~\ref{lem:konig}) to decompose the edges assigned to each block into matchings, forming a schedule for each block, and then concatenates these schedules for $b = 1, \dots, B$.

Write $\IRA^\tau_\lambda(C_j)$ for the finishing time of $E_j$ in this
schedule. We will show that $\IRA_{\lambda}^{\tau}(C_j) \leq P_{\lambda, \tau}(C_j)$ where 
\begin{equation}\label{eq:P-response}
 P_{\lambda,\tau}(x)
 =\lambda+(\tau+1)\max\left\{0,\ceil{\frac{x-\lambda}{\tau}}\right\},
 \qquad x\ge1.
\end{equation}

\begin{lemma}\label{prop:calendar-response}
For deadlines $C_1,\ldots,C_n$ for which \eqref{lp:I} is feasible and
positive integers $\lambda,\tau$, algorithm $\IRA^\tau_\lambda$ constructs
in polynomial time a valid coflow schedule satisfying
\[
 \IRA^\tau_\lambda(C_j)\le P_{\lambda,\tau}(C_j)
 \qquad(j\in[n]).
\]
\end{lemma}

\begin{proof}
By Theorem~\ref{lem:rounding}, the edge-to-block assignment obtained in Step~2 guarantees that the edges assigned to the first block have maximum degree at most $T_1$, while for each block $b\ge2$, the maximum degree is at most $T_b-T_{b-1}+1$.

By K\H{o}nig's theorem (Lemma~\ref{lem:konig}), the edges of each block graph can be partitioned into as many matchings as its maximum degree. This gives a valid schedule using at most
$T_1$ slots for the first block and $T_b-T_{b-1}+1$ slots for each
block $b\ge2$.

Fix $j$, and let $q$ satisfy $T_q=\overline{C}_j$.
Every edge of $E_j$ is scheduled within the first $q$ blocks, and
$\overline{C}_j\ge\lambda+(q-1)\tau$. Hence
\begin{align*}
 \IRA^\tau_\lambda(C_j)
 &\le \overline{C}_j+q-1 \le \overline{C}_j+\frac{\overline{C}_j-\lambda}{\tau}
 =P_{\lambda,\tau}(C_j).
\end{align*}
Since $B\le n$, the block LP has polynomial size. Solving it, applying
iterated rounding, and edge-coloring the block graphs all take
polynomial time.
\end{proof}

We use the deadline-preprocessing trick of Rohwedder and
Schnaars~\cite[Lemma~23]{RS2025full} to obtain a variant
$\KIRA^\tau_\lambda$ of $\IRA^\tau_\lambda$.
For the original fractionally feasible deadlines $C_1,\ldots,C_n$,
the edges $J=\{e\in E:C_{j(e)}<2\}$ form a matching.
Algorithm $\KIRA^\tau_\lambda$ schedules $J$ in slot one and removes
its coflows. It then applies $\IRA^\tau_\lambda$ to the remaining
coflows, if any, with their unchanged deadlines, and shifts the
resulting schedule by one slot. Write $\KIRA^\tau_\lambda(C_j)$ for the finishing time of $E_j$
in this schedule.
Lemma~\ref{prop:calendar-response} immediately yields
the following guarantee.

\begin{lemma}\label{prop:kira-bound}
For deadlines $C_1,\ldots,C_n$ for which \eqref{lp:I} is feasible and
positive integers $\lambda,\tau$, algorithm $\KIRA^\tau_\lambda$
constructs in polynomial time a valid coflow schedule satisfying
\begin{equation}\label{eq:K-response}
 \KIRA^\tau_\lambda(C_j)\le
 \begin{cases}
  1,&C_j<2,\\
  1+P_{\lambda,\tau}(C_j),&C_j\ge2,
 \end{cases}
 \qquad(j\in[n]).
\end{equation}
\end{lemma}
\subsection{Final Bound for Coflow Scheduling without Release Dates}\label{sec:main-no-release-proof}

Compute a single set of deadlines using Lemma~\ref{lem:deadlines}.
Run each of the six algorithms in Table~\ref{tab:no-release-six} on
these deadlines and return a schedule of minimum weighted completion
cost. We use the weights\footnote{The parameter choices in Tables \ref{tab:no-release-six} and \ref{tab:release-seven} were obtained by a computer search that enumerated candidate parameter combinations and compared the resulting approximation bounds. The guarantees for the selected combinations are proved analytically in Appendices A and B.} in the table to show that the performance of the best algorithm is bounded by the average performance.

\begin{table}[htbp]
\caption{The six algorithms (first row) and their weights
(second row).}
\label{tab:no-release-six}
\centering
\begin{tabular}{@{}*{6}{c}@{}}
\toprule
$\IRA^2_3$ & $\IRA^2_4$ & $\IRA^3_4$ & $\IRA^3_3$
 & $\KIRA^2_4$ & $\KIRA^3_5$\\
\midrule
$3/21$ & $2/21$ & $5/21$ & $5/21$ & $1/21$ & $5/21$\\
\bottomrule
\end{tabular}
\end{table}

\begin{proof}[Proof of Theorem~\ref{thm:main-no-release}]
Let $\sigma_1,\ldots,\sigma_6$ be the schedules produced by algorithms $\IRA$ and $\KIRA$ in table order. Let $a_1,\ldots,a_6$ be corresponding weights given by the table.
Recall that $C_j^*(\sigma_i)$ denotes the completion time of the $j$-th coflow in schedule $\sigma_i$.  If schedule $\sigma_i$ is chosen with probability $a_i$, the expected
completion time of coflow $E_j$ is
\[
 D_j=\sum_{i=1}^6 a_i C_j^*(\sigma_i).
\]

\begin{claim}\label{clm:no-release-mean}
    For every $j\in[n]$, $21D_j\le29(C_j+1)$.
\end{claim}

The claim follows from Lemmas~\ref{prop:calendar-response}
and~\ref{prop:kira-bound} and an explicit calculation, given in
Appendix~\ref{app:no-release-claim}.

Let $\sigma$ minimize $\cost$ among the six schedules.
Since $a_i\ge0$ and $\sum_{i=1}^6a_i=1$, the minimum cost is at
most the weighted average:
\begin{align*}
 \cost(\sigma)
 &=\min_{i\in[6]}\cost(\sigma_i)
 \le\sum_{i=1}^6 a_i\cost(\sigma_i)\\
 &=\sum_{j\in[n]}\omega_jD_j
 \le\frac{29}{21}\sum_{j\in[n]}\omega_j(C_j+1)
 \le\frac{58}{21}\OPT,
\end{align*}
where the last inequality is guaranteed by \eqref{eq:deadline-potential}.
The deadline construction and all six scheduling algorithms take polynomial time. 
\end{proof}

\section{Coflow Scheduling with Release Dates}\label{sec:release}

In this section we prove $3.17$ approximation for coflow scheduling with release dates. For this, we define three types of algorithms: $\SARA_H, \LARA_p^{\mathcal L}$, and $\KLARA_p^{\mathcal L}$. All three algorithms receive as an input a set of  deadlines. To compute it, we use a result of \cite{Im2019}, which is described in Section~\ref{sec:fdrd}. All algorithms will also use list coloring lemmas, which are stated in Section~\ref{sec:list-coloring}. Finally, we define algorithms $\SARA$ in Section~\ref{sec:sarah} and $\LARA$ together with $\KLARA$ in Section~\ref{sec:release-blocks}. In Section~\ref{sec:release-proof}, we combine all these algorithms to obtain the desired $3.17$ approximation.

\subsection{Fractional Schedule with Release Dates}
\label{sec:fdrd}
Let $C_1,\ldots,C_n$ be deadlines for coflows $E_1,\ldots,E_n$,
respectively. Set $t_0=0$, and let $t_1<\cdots<t_k$ be the
distinct positive release dates and deadlines, so that $k\le2n$.
We say that the deadlines are \emph{fractionally feasible} if the
following linear program admits a solution $x_{s, e} \in [0, 1]$:
\begingroup
\renewcommand{\theequation}{LP R}
\begin{subequations}\label{lp:R}
\renewcommand{\theequation}{LP R.\arabic{equation}}
\begin{align}
 \sum_{s\in[k]}x_{s,e}&=1
       &&\forall e\in E,
       \label{eq:lp-R-assignment}\\
 \sum_{e\in\delta(v)}x_{s,e}&\le t_s-t_{s-1}
       &&\forall s\in[k],\ \forall v\in V,
       \label{eq:lp-R-capacity}\\
 x_{s,e}&=0
       &&\forall e\in E,\ \forall s\in[k]:
         t_{s-1}<r_e\ \text{or}\ t_s>C_e
       \label{eq:lp-R-window}\\
 x_{s,e}&\ge0
       &&\forall s\in[k],\ \forall e\in E.
       \label{eq:lp-R-nonnegative}
\end{align}
\end{subequations}
\endgroup
Here $x_{s,e}$ is the fraction of the edge $e$ assigned to the block
$(t_{s-1},t_s]$. Constraints~\eqref{eq:lp-R-assignment} ensure that
every edge of each coflow is fully scheduled, and
\eqref{eq:lp-R-capacity} ensure that the fractional degree of each
vertex in every block does not exceed the block size $t_s-t_{s-1}$.
Constraints~\eqref{eq:lp-R-window} forbid edge scheduling before
the coflow's release date or after its deadline.

The deadline construction of Im, Moseley, Pruhs, and
Purohit~\cite{Im2019} also applies with release dates. We use its
deterministic, real-deadline form described by Rohwedder and
Schnaars~\cite[Section~5]{RS2025}.

\begin{lemma}[{Im--Moseley--Pruhs--Purohit~\cite{Im2019}}]
\label{lem:release-deadlines}
In deterministic polynomial time one can compute deadlines
$C_j\ge r_j+1$ for which \eqref{lp:R} is feasible and
\begin{equation}\label{eq:release-deadline-potential}
 \sum_{j\in[n]}\omega_j(C_j+1)\le2\OPT.
\end{equation}
\end{lemma}

Every valid schedule satisfies $C_j^*\ge r_j+1$, so
\begin{equation}\label{eq:release-lower-bound}
 \sum_{j\in[n]}\omega_j(r_j+1)\le\OPT.
\end{equation}




\subsection{List Coloring Lemmas}\label{sec:list-coloring}

In this section we describe list-coloring algorithms that we will later use as a subroutine in algorithms $\SARA_H$, $\LARA_p^{\mathcal L}$ and $\KLARA_p^{\mathcal L}$. A list edge coloring assigns each edge $e$ a color in a specified finite
set $L_e$, with incident edges receiving different colors. We use the
following rank-based formulation. Its proof follows the stable-matching
method of Galvin~\cite{Galvin1995} and the related list-coloring arguments
of Borodin, Kostochka, and Woodall~\cite{BKW1997}.

\begin{lemma}[{List Coloring}]\label{lem:galvin}
Let $G=(V,E)$ be a bipartite multigraph and let bijection $r_v:\delta_G(v)\to[\deg_G(v)]$ specify an ordering of edges incident to vertex $v$. Suppose every edge $e=\{ u, v\} \in E$ has a finite color list satisfying
\begin{equation}\label{eq:rank-list-size}
 |L_e|\ge r_u(e)+r_v(e)-1.
\end{equation}
Then $G$ admits a proper edge coloring with $\operatorname{col}(e)\in L_e$ for all $e\in E$. Such a coloring can be found in polynomial time.
\end{lemma}
\begin{proof}
We argue by induction on the number of edges. The case $E = \varnothing$ is trivial. Fix some color $c$ that occurs in at least one list, and denote by $E_c \subseteq E$ the subset of edges that have this color in the list. For each vertex $v$, $r_v$ defines a preferences order: $e \succ_v e'$ iff $r_v(e) < r_v(e')$. We apply Theorem~\ref{thm:gale-shapley} to graph $(V, E_c)$ to find a stable matching $F \subseteq E_c$. Note that $F \neq \varnothing$ because $E_c \neq \varnothing$. We color each edge in $F$ by $c$ and consider the remaining set of edges $E \setminus F$. For $v \in V$ and $e \in \delta_G(v) \setminus F$, define $r'_v(e) = |\{ e' \in E \setminus F ~|~ r_v(e') < r_v(e), e' \in \delta_{G}(v)\}| + 1$ as a new rank function for the graph where we deleted edges $F$. 

\begin{claim}
    For any $e \in E \setminus F, e = \{u, v\}$, we have $|L_e \setminus \{c\}| \geq r'_u(e) + r_v'(e) - 1$.
\end{claim}
Note that if the claim is true, then we can just proceed by induction on the graph $(V, E \setminus F)$. Let us prove the claim. Consider an edge $e \in E \setminus F$. If $c \notin L_e$, then $L_e \setminus \{c\} = L_e$ and so the bound holds. If $c \in L_e$, then $e$ has not been chosen in the stable matching, implying that there is an edge $e' \in F$, such that either $r_u(e') < r_u(e)$ or $r_v(e') < r_v(e)$. Both cases imply $r'_u(e) + r'_v(e) < r_u(e) + r_v(e)$, proving the claim.

Finally, note that the algorithm works in polynomial time because stable matching can be computed in polynomial time and after each step we reduce the number of edges at least by $1$.
\end{proof}

We use the following consequence of Lemma~\ref{lem:galvin},
which is Galvin's list edge-coloring theorem. It also follows from
the stronger local bound in~\cite[Theorem~3]{BKW1997}.
\begin{corollary}[Galvin~\cite{Galvin1995}]
\label{cor:bounded-degree}
    Let $G=(V,E)$ be a bipartite multigraph with maximum degree $\Delta$. Suppose for each edge, there is a list $L_e$ of size at least $\Delta$. Then $G$ admits a proper edge coloring with $\operatorname{col}(e)\in L_e$ for all $e\in E$. Such a coloring can be found in polynomial time.
\end{corollary}

\begin{proof}
    By Lemma~\ref{lem:konig}, we can partition edges of $G$ into $\Delta$ matchings: $E = E_1 \sqcup \ldots \sqcup E_{\Delta}$.  For an edge $e \in E$, let $p_e$ be the matching in which it lies $e \in E_{p_e}$. Let $V = V_1 \sqcup V_2$ be the parts of $G$. For a vertex $v \in V$, order incident edges $\delta_G(v)$ to it by an increasing value of $p_e$ if $v \in V_1$ and by decreasing value of $p_e$ if $v \in V_2$. Let $e_1, \ldots, e_k$ be the obtained order. Define ranks by setting $\rho_v(e_i) = i$. 

Let us show that the condition of Lemma~\ref{lem:galvin} is satisfied. Let $e \in E, e = \{u, v\}, u \in V_1, v \in V_2$ be an edge. At most $p_e - 1$ edges are before the edge $e$ in the order for vertex $u$. Similarly, at most $\Delta - p_e$ edges are before $e$ in the order for vertex $v$. Therefore, $\rho_u(e) \leq p_e$ and $\rho_v(e) \leq \Delta - p_e + 1$, implying that $\rho_u(e) + \rho_v(e) - 1 \leq \Delta$. Therefore, the desired coloring exists and can be found in polynomial time.
\end{proof}

\subsection{Scheduling via Rank Assignment}
\label{sec:sarah}

In this section we define the scheduling algorithm $\SARA_H$. This algorithm also receives as an input deadlines $C_1 \leq \ldots \leq C_n$, such that \ref{lp:R} is fractionally feasible for them, after relabeling coflows. First, the algorithm constructs ranks that satisfy the following guarantees.

\begin{lemma}
    \label{lem:getranks}
    Let $G=(V_1\sqcup V_2,E)$ be a bipartite multigraph. Suppose for each $e \in E$, there is a deadline $C_e \geq 0$, such that the following holds. For any $v \in V_1 \sqcup V_2$ and $t \geq 0$, there is at most $\lfloor t \rfloor$ edges $e \in E$, such that $v \in e$ and $C_e < t$. Then, for every integer $H \geq 1$, one can construct in polynomial time bijective ranks $\rho_v : \delta_G(v) \to [\deg_G(v)]$, such that for any edge $e = \{u, v\}$, it holds 

    \begin{equation}\label{eq:rank-bound}
 \rho_u(e) + \rho_v(e) - 1 \le
 \begin{cases}
  H,&C_e < H + 1,\\
  2\lfloor C_e \rfloor - 1 &C_e\ge H + 1.
 \end{cases}
\end{equation}
\end{lemma}

\begin{proof}
    Denote $V = V_1 \sqcup V_2$. Let $E' = \{ e \in E ~|~ C_e < H + 1\}$ be the set of edges with deadline before $H + 1$. Note that each vertex is incident to at most $H$ edges from $E'$. By Lemma~\ref{lem:konig}, we can partition them into at most $H$ matchings: $E' = E_1 \sqcup \ldots \sqcup E_H$. For an edge $e \in E'$, let $p_e$ denote the matching to which it belongs $e \in E_{p_e}$. For a vertex $v \in V$, let us order all edges incident to it: $e_1, e_2, \ldots, e_k$ as follows. (1) Edges with $C_e < H + 1$ precede all edges with $C_e \geq H + 1$. (2) Edges with $C_e \geq H + 1$ are ordered by an increasing value of $C_e$. (3) Edges with $C_e < H + 1$ are ordered by an increasing value of $p_e$ if $v \in V_1$ and by a decreasing value of $p_e$ if $v \in V_2$. Define $\rho_v(e)$ as the index of $e$ in this order. Let us show that the constructed ranks $\rho_v$ satisfy the conditions of the lemma. 
    
    Let $e \in E', e = \{u, v\}$, where $u \in V_1, v \in V_2$. Then, the index $\rho_u(e)$ of $e$ in the order for vertex $u$ is at most $p_e$, while the index $\rho_v(e)$ of $e$ in the order for vertex $v$ is at most $H - p_e + 1$. So, $\rho_u(e) + \rho_v(e) - 1 \leq H$, satisfying the condition.

    Now consider the case $e \in E \setminus E'$. In this case there are at most $\lfloor C_e \rfloor - 1$ edges that precede $e$ both in the order for $v$ and $u$. Therefore, the sum $\rho_v(e) + \rho_u(e)$ is at most $2 \lfloor C_e \rfloor$.

Finally, note that König and computing orders for each vertex can be done in polynomial time.
\end{proof}

Recall that for an edge $e \in E$, we set $C_e$ as a deadline of its coflow. Since the deadlines $C_1, \ldots, C_n$ are fractionally feasible, for any vertex $v$ and $t \geq 0$, at most $\lfloor t \rfloor$ edges $e \in E$ are incident to it such that $C_e < t$, because we cannot schedule more than $\lfloor t \rfloor$ edges in a time window of length less than $t$. Therefore, the conditions of Lemma~\ref{lem:getranks} are satisfied. $\SARA_H$ computes ranks $\rho_v$ given by Lemma~\ref{lem:getranks}. Then, $\SARA_H$ defines $L_e = \{ r_e + 1, r_e + 2, \ldots, r_e + \rho_u(e) + \rho_v(e)- 1\}$ as a list of possible colors for an edge $e \in E, e = \{u, v\}$. It then applies Lemma~\ref{lem:galvin} to lists $L_e$ to get a coloring $c : E \to \mathbb{N}$ with a guarantee that $c(e) \in L_e$. The algorithm assigns edge $e$ to slot $c(e)$. By construction, this is a valid schedule since every vertex is incident to at most one edge of a given color. Let $\SARA_H(C_j, r_j)$ denote the finishing time of coflow $E_j$. Let us formally state the guarantee that the algorithm has:

\begin{lemma}
\label{lem:sara-bound}
For deadlines $C_1 \leq \ldots \leq C_n$ for which \eqref{lp:R} is feasible and any integer
$H\ge1$, algorithm $\SARA_H$ constructs in deterministic polynomial
time a valid schedule satisfying, for every $j\in[n]$,
\begin{equation}\label{eq:sarah-bound}
\SARA_H(C_j,r_j)\le r_j+
\begin{cases}
H,&C_j<H+1,\\
2\floor{C_j}-1,&C_j\ge H+1.
\end{cases}
\end{equation}
\end{lemma}
\begin{proof}
The lists satisfy Lemma~\ref{lem:galvin}, and every listed slot is
after the edge's release date. The bound follows from \eqref{eq:rank-bound}.

Each list has at most $2|E|-1$ slots because the ranks are positions among incident edges. Thus the total list size is $O(|E|^2)$. Matching search in Lemma~\ref{lem:galvin} works in polynomial time. Therefore, the whole procedure is polynomial.
\end{proof}

\subsection{Scheduling Algorithms LARA and KLARA}\label{sec:release-blocks}

Algorithms $\LARA_p^{\mathcal L}$ and $\KLARA_p^{\mathcal L}$
take a coflow instance with integer release dates, fractionally
feasible deadlines $C_1\le\cdots\le C_n$ and two additional parameters: a positive integer $p$,
and a nonempty list $\mathcal L=(l_1,\ldots,l_s)$ of positive integers.
Let $x$ be a feasible solution of \eqref{lp:R}, with LP boundaries
$0=t_0<t_1<\cdots<t_k$. The last entry of $\mathcal L$ repeats:
set $\ell_b=l_{\min\{b,s\}}$ for every $b\ge1$.

\textbf{1. Blocks and Fractional Assignment.}
Set $T_0=0$ and $T_b=\sum_{h=1}^b\ell_h$. Choose the smallest
$B\ge\max\{p,s\}$ with $T_B>C_n$, and form blocks $(T_0, T_1] \sqcup (T_1, T_2] \sqcup \ldots \sqcup (T_{B-1}, T_B]$.
One feasible fractional assignment to these blocks is

\begin{equation}\label{eq:release-fractional-assignment}
 z_{b,e}=\sum_{s=1}^{k}
   \frac{\lvert (t_{s-1},t_s]\cap(T_{b-1},T_b]\rvert}{\lvert (t_{s-1}, t_s]\rvert}
   x_{s,e},
 \qquad b\in[B],\ e\in E.
\end{equation}
Intuitively, $z_{b, e}$ denotes a fraction of the edge that is scheduled in block $b$ according to the fractional schedule $x_{s, e}$. A slightly relaxed version of condition \ref{eq:lp-R-window} holds for $z_{b, e}$:

\begin{claim}\label{claim:losebound}
For every $e\in E$ and $b\in[B]$,
\[
z_{b,e}=0\quad\text{if }T_b\le r_e\ \text{or}\ T_{b-1}\ge C_e.
\]
\end{claim}
\begin{proof}
In either case, the block $(T_{b-1},T_b]$ has no positive-length
intersection with the processing window $(r_e,C_e]$.
\end{proof}

\textbf{2. Rounding.} Next, we round the fractional schedule $z_{b, e}$. We apply Theorem~\ref{lem:rounding} with
\[
 \mathcal B=[B],\qquad
 \mathcal C=V\times[B],\qquad
 c_{v,b}=T_b-T_{b-1},
\]
and
\begin{equation}\label{eq:window-support}
 \mathcal B_e=\{b\in[B]:r_e<T_b\ \text{and}\ T_{b-1}<C_e\}.
\end{equation}
Note that $\calB_e$ is defined as a subset of blocks having non-empty intersection with $(r_e,C_e]$. The assignment $z$ satisfies \eqref{eq:ir-assign}--\eqref{eq:ir-capacity}. Applying
Theorem~\ref{lem:rounding} with protected block $p$ assigns every
edge to one block $b\in\mathcal B_e$, while violating each
vertex--block capacity by at most one for $b\ne p$ and not violating
it for $b=p$. Let $F_b\subseteq E$ be the edges assigned to block $b$.
Then
\begin{equation}\label{eq:release-degree}
 \Delta(F_b)\le d_b,
 \qquad d_b=\ell_b+1-\ind_{\{b=p\}}.
\end{equation}

\textbf{3. Scheduling via List Edge Coloring.}

Let us process blocks $b=1,\ldots,B$ in order. $S_b$ will upper bound the maximum finishing time of an edge in block $b$. Initially, $S_0 = 0$.
For every edge $e\in F_b$, form the list
\begin{equation}\label{eq:interval-lists}
 L_e=\{\max\{S_{b-1},r_e\}+1,\ldots,
        \max\{S_{b-1},r_e\}+d_b\}.
\end{equation}
The subgraph $(V, F_b)$ has degree at most $d_b$. Therefore, we can apply Corollary~\ref{cor:bounded-degree} to properly
color the edges of $F_b$ from these lists, and schedule each edge
in the slot given by its color. Therefore, the upper bound for block $b$ can be set as
\begin{equation}\label{eq:release-recurrence}
 S_b=\max\{S_{b-1},T_b-1\}+d_b,
\end{equation}
since edges $e \in F_b$ satisfy $r_e<T_b$ by
\eqref{eq:window-support}, and both $r_e$ and $T_b$ are integers.

Write $\LARA^{\mathcal{L}}_p(C_j, r_j)$ for the finishing time of $E_j$ in this schedule. We have the following bound on the schedule:

\textbf{KLARA Algorithm.}

Using the same preprocessing as $\KIRA$, we define a variant
$\KLARA$ of $\LARA$. $\KLARA$ will act as $\LARA$ but with one more phase in the beginning. It takes edges $J=\{e\in E:C_{j(e)}<2\}$ with deadline less than $2$ and schedules them in the first slot. $J$ forms a matching because otherwise it is impossible to schedule two edges sharing a vertex in window of size less than $2$, even fractionally. Note also any edge $e \in J$ should have release date $r_e = 0$ because $r_e + 1 \leq C_e < 2$ and $r_e$ is integer. $\KLARA$ schedules $J$ in slot one and removes its
coflows. After that, it executes steps $1, 2, 3$ on the remaining instance. Note that we have the same guarantees as for $\LARA$, since the removed part $J$ only simplifies the problem. In the third step, we initialize $S_0 = 1$ instead of $S_0 = 0$, since the first slot is already engaged.

Write $\KLARA^{\mathcal{L}}_p(C_j, r_j)$ for the resulting finishing time of $E_j$. 

Let us state the final guarantees that we have on the finishing time for $\LARA$ and $\KLARA$:

\begin{lemma}
    Let $C_1 \leq \ldots \leq C_n$ be deadlines for which \ref{lp:R} is feasible after relabeling the coflows. For a sequence of natural numbers $\mathcal{L} = (\ell_1, \ldots, \ell_b)$ and a natural number $p$, define 
    $$
    T_b = T_{b - 1} + \ell_b, \quad d_b = \ell_b + 1 - \mathbf1_{b = p},  \quad  S_b^q = \max \{ S^q_{b-1}, T_b - 1  \} + d_b
    $$
    with $T_0 = 0$ and $S_0^q = q$ for $q \in \{0, 1\}$.
    
   Let $b \in [B]$ be such that $C_j \in (T_{b-1}, T_b]$.  Then, 
    \begin{equation}\label{eq:coflow-local-bound}
 \LARA^{\mathcal{L}}_p(C_j, r_j)\le\max\{S^0_{b-1},r_j\}+d_b,
\end{equation}
    and 
    \begin{equation}\label{eq:klara-local-bound}
 \KLARA^{\mathcal{L}}_p(C_j, r_j)\le
 \begin{cases}
  1,&C_j<2,\\
  \max\{S^1_{b-1},r_j\}+d_b,&C_j\ge2.
 \end{cases}
\end{equation}
\end{lemma}

\begin{proof}
    Let $e \in E_j$ be an edge. By Claim~\ref{claim:losebound}, this edge lies in block $b$ or an earlier one. Therefore, it is scheduled by $\max\{S_{b-1}^0, r_j\} + d_b$ for $\LARA$ and by $\max\{S_{b-1}^1, r_j\} + d_b$ for $\KLARA$. Moreover, if $C_j < 2$, it is scheduled in the first slot in case of $\KLARA$.
\end{proof}

Finally, let us remark that both algorithms can be implemented in polynomial time. Note that the only issue is that release dates can take exponentially large values. 

Assume we have a release date $r_j$, such that for any other release date $r_i < r_j$, it holds that $r_i + |E| < r_j$. Then, without loss of generality we may assume that for release dates $r_i < r_j$ all edges are scheduled before $r_j$. If it does not hold, then we can take an edge that is scheduled after $r_j$ and move it to an earlier slot (there should be a free slot since we have at most $|E|$ edges). Therefore, our instance partitions into two independent instances: with release dates $< r_j$ and with release dates $\geq r_j$. So, w.l.o.g. we can assume that $|r_j - r_i| \leq n \cdot (|E| + 1)$. 

The slots before the minimal release date $\min_{i=1}^n r_i$ are not used. Thus, we can only consider slots after $\min_{i=1}^n r_i$. There are polynomially many of them, implying that the whole algorithm can be implemented in polynomial time.

\subsection{Final Bound for Coflow Scheduling with Release Dates}
\label{sec:release-proof}\label{sec:release-calendars}

The final procedure works as follows. We compute deadlines using
Lemma~\ref{lem:release-deadlines} and relabel the coflows so that
$C_1\le\cdots\le C_n$. We then run each of the seven algorithms in
Table~\ref{tab:release-seven} on these deadlines and return a schedule
of minimum weighted completion cost. The weights in the table are
used only to bound the minimum cost by the weighted average.

\begin{table}[htbp]
\caption{The seven scheduling algorithms: indices $i$ (first row),
algorithms (second row), and analysis weights $a_i$ (third row).
The last entry of each block-length tuple $\mathcal{L}$ of $\LARA_p^{\mathcal L}$ and $\KLARA_p^{\mathcal L}$ repeats indefinitely.}
\label{tab:release-seven}\label{tab:release-six}
\centering
\small
\setlength{\tabcolsep}{4pt}
\begin{tabular}{@{}*{7}{c}@{}}
\toprule
$1$ & $2$ & $3$ & $4$ & $5$ & $6$ & $7$\\
\midrule
$\SARA_3$ & $\SARA_4$ & $\SARA_5$
& $\LARA_2^{(1,2)}$
& $\LARA_2^{(1,3,2)}$
& $\LARA_2^{(1,2,2,1)}$
& $\KLARA_1^{(2)}$\\
\midrule
$16/575$ & $2/575$ & $3/575$ & $258/575$
& $74/575$ & $42/575$ & $180/575$\\
\bottomrule
\end{tabular}
\end{table}

\begin{proof}[Proof of Theorem~\ref{thm:main-release}]
Let $\sigma_i$ be the schedule produced by the algorithm in column $i$
of Table~\ref{tab:release-seven}, with corresponding weight $a_i$.
If schedule $\sigma_i$ is chosen with probability $a_i$, the expected
completion time of coflow $E_j$ is
\[
 D_j=\sum_{i=1}^7 a_i C_j^*(\sigma_i).
\]
We have the following bound.

\begin{claim}\label{clm:release-mean}
For every $j\in[n]$, $575D_j\le900(C_j+1)+21(r_j+1)$.
\end{claim}

The claim follows from Lemma~\ref{lem:sara-bound},
\eqref{eq:coflow-local-bound}, \eqref{eq:klara-local-bound}, and the
parameters in Table~\ref{tab:release-seven}. Its proof is given in
Appendix~\ref{app:release-claim}.

Let $\sigma$ be the schedule returned. Since $a_i\ge0$ and
$\sum_{i=1}^7a_i=1$, we obtain
\begin{align*}
 \cost(\sigma)
 &=\min_{i\in[7]}\cost(\sigma_i)
 \le\sum_{i=1}^7 a_i\cost(\sigma_i)\\
 &=\sum_{j\in[n]}\omega_jD_j\\
 &\le\frac{900}{575}\sum_{j\in[n]}\omega_j(C_j+1)
     +\frac{21}{575}\sum_{j\in[n]}\omega_j(r_j+1)\\
 &\le\left(2\cdot\frac{900}{575}+\frac{21}{575}\right)\OPT
 =\frac{1821}{575}\OPT,
\end{align*}
where the last inequality follows from
\eqref{eq:release-deadline-potential} and~\eqref{eq:release-lower-bound}.
The deadline construction, all seven scheduling algorithms, and the
final cost comparison take polynomial time. Fixed tie-breaking
makes the procedure deterministic.
\end{proof}

\bibliographystyle{ACM-Reference-Format}
\bibliography{references}

@article{Mastrolilli2010,
author = {Monaldo Mastrolilli and Maurice Queyranne
and Andreas S. Schulz and Ola Svensson and Nelson A. Uhan},
title = {Minimizing the Sum of Weighted Completion Times
in a Concurrent Open Shop},
journal = {Operations Research Letters},
year = {2010},
volume = {38},
number = {5},
pages = {390--395},
doi = {10.1016/j.orl.2010.04.011}
}

@inproceedings{SachdevaSaket2013,
author = {Sushant Sachdeva and Rishi Saket},
title = {Optimal Inapproximability for Scheduling Problems via
Structural Hardness for Hypergraph Vertex Cover},
booktitle = {2013 IEEE Conference on Computational Complexity},
year = {2013},
pages = {219--229},
publisher = {IEEE},
doi = {10.1109/CCC.2013.30}
}

@misc{ImPurohit2018,
author = {Sungjin Im and Manish Purohit},
title = {A Tight Approximation for Co-flow Scheduling
for Minimizing Total Weighted Completion Time},
year = {2018},
eprint = {1707.04331},
archivePrefix = {arXiv},
primaryClass = {cs.DS},
doi = {10.48550/arXiv.1707.04331},
url = {https://arxiv.org/abs/1707.04331v2},
note = {Version 2. Includes an authors' note identifying
an error in the claimed approximation guarantee}
}

@inproceedings{Malewicz2010,
author = {Grzegorz Malewicz and Matthew H. Austern and Aart J. C. Bik
and James C. Dehnert and Ilan Horn and Naty Leiser
and Grzegorz Czajkowski},
title = {{Pregel}: A System for Large-Scale Graph Processing},
booktitle = {Proceedings of the 2010 ACM SIGMOD International
Conference on Management of Data},
year = {2010},
pages = {135--146},
publisher = {ACM},
doi = {10.1145/1807167.1807184}
}

@inproceedings{Agarwal2018,
  author    = {Agarwal, Saksham and Rajakrishnan, Shijin and Narayan, Akshay
               and Agarwal, Rachit and Shmoys, David and Vahdat, Amin},
  title     = {{Sincronia: Near-optimal network design for coflows}},
  booktitle = {Proceedings of the 2018 Conference of the ACM Special Interest
               Group on Data Communication},
  series    = {SIGCOMM '18},
  year      = {2018},
  pages     = {16--29},
  publisher = {Association for Computing Machinery},
  address   = {New York, NY, USA},
  doi       = {10.1145/3230543.3230569}
}

@inproceedings{Ahmadi2017,
  author    = {Ahmadi, Saba and Khuller, Samir and Purohit, Manish and Yang, Sheng},
  title     = {{On Scheduling Coflows}},
  booktitle = {Integer Programming and Combinatorial Optimization (IPCO 2017)},
  editor    = {Eisenbrand, Friedrich and Koenemann, Jochen},
  series    = {Lecture Notes in Computer Science},
  volume    = {10328},
  year      = {2017},
  pages     = {13--24},
  publisher = {Springer International Publishing},
  address   = {Cham, Switzerland},
  doi       = {10.1007/978-3-319-59250-3_2}
}

@book{AignerZiegler2018,
  author    = {Aigner, Martin and Ziegler, G{\"u}nter M.},
  title     = {{Proofs from THE BOOK}},
  edition   = {6th},
  year      = {2018},
  publisher = {Springer},
  address   = {Berlin, Heidelberg},
  isbn      = {978-3-662-57264-1},
  doi       = {10.1007/978-3-662-57265-8},
  note      = {Chapter~38, ``The Dinitz problem'', pp.~271--276}
}

@article{BeckFiala1981,
  author  = {Beck, J{\'o}zsef and Fiala, Tibor},
  title   = {{``Integer-making'' theorems}},
  journal = {Discrete Applied Mathematics},
  year    = {1981},
  volume  = {3},
  number  = {1},
  pages   = {1--8},
  doi     = {10.1016/0166-218X(81)90022-6}
}

@article{BKW1997,
  author  = {Borodin, Oleg V. and Kostochka, Alexandr V. and Woodall, Douglas R.},
  title   = {{List Edge and List Total Colourings of Multigraphs}},
  journal = {Journal of Combinatorial Theory, Series B},
  year    = {1997},
  volume  = {71},
  number  = {2},
  pages   = {184--204},
  doi     = {10.1006/jctb.1997.1780}
}

@inproceedings{Chowdhury2012,
  author    = {Chowdhury, Mosharaf and Stoica, Ion},
  title     = {{Coflow: A networking abstraction for cluster applications}},
  booktitle = {Proceedings of the 11th ACM Workshop on Hot Topics in Networks},
  series    = {HotNets-XI},
  year      = {2012},
  pages     = {31--36},
  publisher = {Association for Computing Machinery},
  address   = {New York, NY, USA},
  doi       = {10.1145/2390231.2390237}
}

@inproceedings{Chowdhury2014,
  author    = {Chowdhury, Mosharaf and Zhong, Yuan and Stoica, Ion},
  title     = {{Efficient coflow scheduling with Varys}},
  booktitle = {Proceedings of the 2014 ACM Conference on SIGCOMM},
  series    = {SIGCOMM '14},
  year      = {2014},
  pages     = {443--454},
  publisher = {Association for Computing Machinery},
  address   = {New York, NY, USA},
  doi       = {10.1145/2619239.2626315}
}

@inproceedings{Fukunaga2022,
  author    = {Fukunaga, Takuro},
  title     = {{Integrality Gap of Time-Indexed Linear Programming Relaxation
               for Coflow Scheduling}},
  booktitle = {Approximation, Randomization, and Combinatorial Optimization.
               Algorithms and Techniques (APPROX/RANDOM 2022)},
  editor    = {Chakrabarti, Amit and Swamy, Chaitanya},
  series    = {Leibniz International Proceedings in Informatics (LIPIcs)},
  volume    = {245},
  year      = {2022},
  pages     = {36:1--36:13},
  publisher = {Schloss Dagstuhl -- Leibniz-Zentrum f{\"u}r Informatik},
  address   = {Dagstuhl, Germany},
  isbn      = {978-3-95977-249-5},
  doi       = {10.4230/LIPIcs.APPROX/RANDOM.2022.36}
}

@article{GaleShapley1962,
  author  = {Gale, David and Shapley, Lloyd S.},
  title   = {{College Admissions and the Stability of Marriage}},
  journal = {The American Mathematical Monthly},
  year    = {1962},
  volume  = {69},
  number  = {1},
  pages   = {9--15},
  doi     = {10.2307/2312726}
}

@article{Galvin1995,
  author  = {Galvin, Fred},
  title   = {{The List Chromatic Index of a Bipartite Multigraph}},
  journal = {Journal of Combinatorial Theory, Series B},
  year    = {1995},
  volume  = {63},
  number  = {1},
  pages   = {153--158},
  doi     = {10.1006/jctb.1995.1011}
}

@inproceedings{Im2019,
  author    = {Im, Sungjin and Moseley, Benjamin and Pruhs, Kirk and Purohit, Manish},
  title     = {{Matroid Coflow Scheduling}},
  booktitle = {46th International Colloquium on Automata, Languages, and
               Programming (ICALP 2019)},
  editor    = {Baier, Christel and Chatzigiannakis, Ioannis and Flocchini, Paola
               and Leonardi, Stefano},
  series    = {Leibniz International Proceedings in Informatics (LIPIcs)},
  volume    = {132},
  year      = {2019},
  pages     = {145:1--145:13},
  publisher = {Schloss Dagstuhl -- Leibniz-Zentrum f{\"u}r Informatik},
  address   = {Dagstuhl, Germany},
  isbn      = {978-3-95977-109-2},
  doi       = {10.4230/LIPIcs.ICALP.2019.145}
}

@inproceedings{KhullerPurohit2016,
  author    = {Khuller, Samir and Purohit, Manish},
  title     = {{Brief Announcement: Improved Approximation Algorithms
               for Scheduling Co-Flows}},
  booktitle = {Proceedings of the 28th ACM Symposium on Parallelism
               in Algorithms and Architectures},
  series    = {SPAA '16},
  year      = {2016},
  pages     = {239--240},
  publisher = {Association for Computing Machinery},
  address   = {New York, NY, USA},
  doi       = {10.1145/2935764.2935809}
}

@article{Konig1916,
  author  = {K{\H{o}}nig, D{\'e}nes},
  title   = {{{\"U}ber Graphen und ihre Anwendung auf Determinantentheorie
             und Mengenlehre}},
  journal = {Mathematische Annalen},
  year    = {1916},
  volume  = {77},
  pages   = {453--465},
  doi     = {10.1007/BF01456961}
}

@book{LRS2011,
  author    = {Lau, Lap Chi and Ravi, R. and Singh, Mohit},
  title     = {{Iterative Methods in Combinatorial Optimization}},
  series    = {Cambridge Texts in Applied Mathematics},
  year      = {2011},
  publisher = {Cambridge University Press},
  address   = {Cambridge, UK},
  isbn      = {978-1-107-00751-2},
  doi       = {10.1017/CBO9780511977152}
}

@inproceedings{Qiu2015,
  author    = {Qiu, Zhen and Stein, Cliff and Zhong, Yuan},
  title     = {{Minimizing the Total Weighted Completion Time of Coflows
               in Datacenter Networks}},
  booktitle = {Proceedings of the 27th ACM Symposium on Parallelism
               in Algorithms and Architectures},
  series    = {SPAA '15},
  year      = {2015},
  pages     = {294--303},
  publisher = {Association for Computing Machinery},
  address   = {New York, NY, USA},
  doi       = {10.1145/2755573.2755592}
}

@inproceedings{RS2025,
  author    = {Rohwedder, Lars and Schnaars, Leander},
  title     = {{3.415-Approximation for Coflow Scheduling via Iterated Rounding}},
  booktitle = {52nd International Colloquium on Automata, Languages, and
               Programming (ICALP 2025)},
  editor    = {Censor-Hillel, Keren and Grandoni, Fabrizio and Ouaknine, Jo{\"e}l
               and Puppis, Gabriele},
  series    = {Leibniz International Proceedings in Informatics (LIPIcs)},
  volume    = {334},
  year      = {2025},
  pages     = {128:1--128:19},
  publisher = {Schloss Dagstuhl -- Leibniz-Zentrum f{\"u}r Informatik},
  address   = {Dagstuhl, Germany},
  isbn      = {978-3-95977-372-0},
  doi       = {10.4230/LIPIcs.ICALP.2025.128}
}

@misc{RS2025full,
  author        = {Rohwedder, Lars and Schnaars, Leander},
  title         = {{3.415-Approximation for Coflow Scheduling via Iterated Rounding}},
  year          = {2025},
  eprint        = {2502.21197},
  archivePrefix = {arXiv},
  primaryClass  = {cs.DS},
  doi           = {10.48550/arXiv.2502.21197},
  note          = {Full version of the ICALP 2025 paper}
}

@inproceedings{RS2026,
  author    = {Rohwedder, Lars and Schnaars, Leander},
  title     = {{Graph Scheduling with Group Completion Times}},
  booktitle = {34th Annual European Symposium on Algorithms (ESA 2026)},
  editor    = {Bille, Philip and Pettie, Seth and Storandt, Sabine},
  series    = {Leibniz International Proceedings in Informatics (LIPIcs)},
  volume    = {388},
  year      = {2026},
  pages     = {100:1--100:22},
  publisher = {Schloss Dagstuhl -- Leibniz-Zentrum f{\"u}r Informatik},
  address   = {Dagstuhl, Germany},
  isbn      = {978-3-95977-445-1},
  doi       = {10.4230/LIPIcs.ESA.2026.100}
}

@article{ShafieeGhaderi2017,
  author  = {Shafiee, Mehrnoosh and Ghaderi, Javad},
  title   = {{Scheduling Coflows in Datacenter Networks: Improved Bound
             for Total Weighted Completion Time}},
  journal = {ACM SIGMETRICS Performance Evaluation Review},
  year    = {2017},
  volume  = {45},
  number  = {1},
  pages   = {29--30},
  doi     = {10.1145/3143314.3078548},
  note    = {SIGMETRICS 2017 extended abstract}
}

@article{ShafieeGhaderi2018,
  author        = {Shafiee, Mehrnoosh and Ghaderi, Javad},
  title         = {{An Improved Bound for Minimizing the Total Weighted
                   Completion Time of Coflows in Datacenters}},
  journal       = {IEEE/ACM Transactions on Networking},
  year          = {2018},
  volume        = {26},
  number        = {4},
  pages         = {1674--1687},
  doi           = {10.1109/TNET.2018.2845852},
  eprint        = {1704.08357},
  archivePrefix = {arXiv},
  primaryClass  = {cs.DS}
}

\appendix
\section{Proof of the Completion-Time Bound}
\label{app:no-release-claim}

\begin{proof}[Proof of Claim~\ref{clm:no-release-mean}]
Fix $j\in[n]$. Recall that
\[
 D_j=\sum_{i=1}^6 a_i C_j^*(\sigma_i),
 \qquad
 (a_1,\ldots,a_6)=\frac1{21}(3,2,5,5,1,5).
\]
We apply Lemmas~\ref{prop:calendar-response} and~\ref{prop:kira-bound}
and distinguish three cases.

If $1\le C_j<2$, the six completion-time bounds are
$3,4,4,3,1,1$, respectively. Therefore,
\[
 21D_j
 \le3\cdot3+2\cdot4+5\cdot4+5\cdot3+1+5
 =58\le29(C_j+1).
\]
If $C_j=2$, the last two bounds become $5$ and $6$, while the first
four remain unchanged. Hence
\[
 21D_j
 \le3\cdot3+2\cdot4+5\cdot4+5\cdot3+5+5\cdot6
 =87=29(C_j+1).
\]

Suppose now that $C_j>2$, and put $m=\ceil{C_j}$.
The bounds for the six schedules give
\begin{align*}
 21D_j
 &\le3P_{3,2}(C_j)+2P_{4,2}(C_j)
      +5P_{4,3}(C_j)+5P_{3,3}(C_j)\\
 &\quad+\bigl(1+P_{4,2}(C_j)\bigr)
      +5\bigl(1+P_{5,3}(C_j)\bigr)\\
 &=3\bigl(P_{3,2}(C_j)+P_{4,2}(C_j)\bigr)\\
 &\quad+5\bigl(P_{3,3}(C_j)+P_{4,3}(C_j)
              +P_{5,3}(C_j)\bigr)+6.
\end{align*}
The constant $6=1+5$ accounts for the one-slot shifts in the two
$\KIRA$ schedules.

To evaluate the two sums, we use
\[
 \sum_{h=0}^{t-1}\ceil{\frac{z-h}{t}}=\ceil{z}
 \qquad(z\in\mathbb{R},\ t\in\mathbb{Z}_{\ge1}).
\]
Indeed, writing $\ceil{z}=qt+r$ with integers $q,r$ and $0\le r<t$,
the summands equal $q+1$ for $h<r$ and $q$ otherwise.

Since $C_j>2$, all ceiling terms in the five relevant $P$-functions
are nonnegative, so the maxima with zero can be omitted.
Applying the identity with $z=C_j-3$ and $t\in\{2,3\}$ yields
\begin{align*}
 P_{3,2}(C_j)+P_{4,2}(C_j)
 &=7+3\sum_{h=0}^{1}\ceil{\frac{C_j-3-h}{2}}\\
 &=7+3(m-3)=3m-2,\\
 P_{3,3}(C_j)+P_{4,3}(C_j)+P_{5,3}(C_j)
 &=12+4\sum_{h=0}^{2}\ceil{\frac{C_j-3-h}{3}}\\
 &=12+4(m-3)=4m.
\end{align*}
Consequently,
\[
 21D_j\le3(3m-2)+5(4m)+6
 =29m\le29(C_j+1),
\]
which completes the proof.
\end{proof}

\section{Proof of the Completion-Time Bound with Release Dates}
\label{app:release-claim}

\begin{proof}[Proof of Claim~\ref{clm:release-mean}]
Fix $j\in[n]$ and put
\begin{equation}\label{eq:release-m}
 m=\floor{C_j}\ge1,
 \qquad 0\le r_j\le m-1,
\end{equation}
where the latter inequality follows from $C_j\ge r_j+1$ and the
integrality of $r_j$. Lemma~\ref{lem:sara-bound} gives
\[
 C_j^*(\sigma_i)\le r_j+h_{i+2}(m)
 \qquad(i\in\{1,2,3\}),
\]
where, for $k\in\{3,4,5\}$,
\[
 h_k(m)=
 \begin{cases}
  k,&1\le m\le k,\\
  2m-1,&m\ge k+1.
 \end{cases}
\]
For $i\in\{4,5,6,7\}$, let $U_i(m)$ be the bound listed for
$\sigma_i$ in Table~\ref{tab:release-responses}. We first show that
$C_j^*(\sigma_i)\le U_i(m)$.

\begin{table}[htbp]
\caption{Upper bounds $U_i(m)$ on $C_j^*(\sigma_i)$ for the four
block schedules in Table~\ref{tab:release-seven}, where
$m=\floor{C_j}$.}
\label{tab:release-responses}
\centering
\begin{tabular}{@{}ccccc@{}}
\toprule
Schedule & $m=1$ & $m=2$ & $m=3$ & $m\ge4$\\
\midrule
$\sigma_4$ & $4$ & $4$ & $7$ & $1+3\ceil{m/2}$\\
$\sigma_5$ & $5$ & $5$ & $5$ & $3+3\floor{m/2}$\\
$\sigma_6$ & $4$ & $4$ & $7$ & $2m-1$\\
$\sigma_7$ & $1$ & $6$ & $6$ & $3+3\floor{m/2}$\\
\bottomrule
\end{tabular}
\end{table}

For each fixed schedule, $T_b$ and $d_b$ refer to its own block
parameters. We write $S_b=S_b^0$ for $\sigma_4,\sigma_5,\sigma_6$
and $S_b=S_b^1$ for $\sigma_7$. Substituting the initial block lengths into
\eqref{eq:release-recurrence} and then applying induction gives
\begin{equation}\label{eq:virtual-endpoints}
\begin{array}{c|l}
\sigma_4 & S_1=2,\quad S_b=3b-2\ (b\ge2),\\
\sigma_5 & S_1=2,\quad S_b=3b\ (b\ge2),\\
\sigma_6 & S_1=2,\quad S_2=4,\quad S_b=2b+1\ (b\ge3),\\
\sigma_7 & S_0=1,\quad S_b=3b\ (b\ge1).
\end{array}
\end{equation}
For $\sigma_7$, the case $m=1$ is handled by the initial matching,
so $U_7(1)=1$.
In every other case, let $k$ be the first index with $T_k>m$.
Since the boundaries are integers, $T_{k-1}\le C_j<T_k$.
If $C_j=T_{k-1}$, all edges of $E_j$ finish by $S_{k-1}$;
otherwise apply \eqref{eq:coflow-local-bound} for
$\sigma_4,\sigma_5,\sigma_6$, or \eqref{eq:klara-local-bound} for
$\sigma_7$. Thus, in either case,
\[
 C_j^*(\sigma_i)\le\max\{S_{k-1},r_j\}+d_k.
\]
The indices below and \eqref{eq:virtual-endpoints} give
$S_{k-1}\ge m-1\ge r_j$, so this bound reduces to $S_{k-1}+d_k$.

For $\sigma_4$, $k=\ceil{m/2}+1$. The bound is $4$ for $m=1,2$
and $3k-2=1+3\ceil{m/2}$ otherwise.
For $\sigma_5$, $m=1,2,3$ gives $k=2$ and
$S_1+d_2=2+3=5$. For $m\ge4$, $k=\floor{m/2}+1$ gives
$3k=3+3\floor{m/2}$.
For $\sigma_6$, $m=1,2,3,4$ gives bounds $4,4,7,7$;
for $m\ge5$, $k=m-1$ gives $2k+1=2m-1$.
For $\sigma_7$ and $m\ge2$, $k=\floor{m/2}+1$ gives
$3k=3+3\floor{m/2}$. This proves all bounds in
Table~\ref{tab:release-responses}, including at integral deadlines.

Set
\begin{align*}
 Q(m)={}&16h_3(m)+2h_4(m)+3h_5(m)\\
       &+258U_4(m)+74U_5(m)+42U_6(m)+180U_7(m).
\end{align*}
We show that
\begin{equation}\label{eq:release-certificate}
 Q(m)\le900m+921
 \qquad(m\ge1).
\end{equation}
For $m=1,2,3,4,5$, direct substitution gives
\[
 \begin{array}{c|rrrrr}
  m&1&2&3&4&5\\
  \hline
  Q(m)&1821&2721&3621&4521&5421
 \end{array}
\]
and hence $Q(m)=900m+921$ in these cases.
For $m\ge6$, all three $h$-functions equal $2m-1$. Writing
$\varepsilon_m=0$ for even $m$ and $\varepsilon_m=1$ for odd $m$,
we obtain
\begin{align*}
 Q(m)
 &=21(2m-1)+258(1+3\ceil{m/2})\\
 &\quad+254(3+3\floor{m/2})+42(2m-1)\\
 &=894m+957+6\varepsilon_m.
\end{align*}
Since $m-\varepsilon_m\ge6$ for every integer $m\ge6$,
\[
 900m+921-Q(m)=6(m-6-\varepsilon_m)\ge0,
\]
proving \eqref{eq:release-certificate}.

Finally, the weights in Table~\ref{tab:release-seven} give
\begin{align*}
 575D_j
 &\le(16+2+3)r_j+Q(m)\\
 &\le21r_j+900m+921\\
 &=900(m+1)+21(r_j+1)\\
 &\le900(C_j+1)+21(r_j+1).
\end{align*}
\end{proof}
\end{document}